\documentclass{amsproc}
\usepackage{hyperref}
\newtheorem{thm}{Theorem}[section]
\newtheorem{cor}[thm]{Corollary}
\newtheorem{lem}[thm]{Lemma}

\theoremstyle{definition}
\newtheorem{defn}[thm]{Definition}
\theoremstyle{remark}
\newtheorem{rem}[thm]{Remark}
\numberwithin{equation}{section}

\usepackage{tikz}
\usepackage{color}
\usepackage{amssymb}
\usepackage[margin=1in]{geometry}
\begin{document}
\title{An elliptic Garnier system}

\author{Chris M. Ormerod}
\address{Caltech, 1200 E California Blvd, Pasadena CA 91125}
\email{cormerod@caltech.edu}
\author{Eric M. Rains}
\address{Caltech, 1200 E California Blvd, Pasadena CA 91125}
\email{rains@caltech.edu}

\begin{abstract}
We present a linear system of difference equations whose entries are expressed in terms of theta functions. This linear system is singular at $4m+12$ points for $m \geq 1$, which appear in pairs due to a symmetry condition. We parameterize this linear system in terms a set of kernels at the singular points. We regard the system of discrete isomonodromic deformations as an elliptic analogue of the Garnier system. We identify the special case in which $m=1$ with the elliptic Painlev\'e equation, hence, this work provides an explicit form and Lax pair for the elliptic Painlev\'e equation. 
\end{abstract}

\maketitle

\tableofcontents

\section{Introduction}

The nonlinear differential equations governing isomonodromic deformations of linear systems of ordinary differential equations play an important role in the theory of integrable systems. The Garnier system is a nonlinear system of commuting ordinary differential equations characterizing the isomonodromic deformations of a second order Fuchsian system with $N+3$ simple poles, three of which are fixed at $0$, $1$ and $\infty$, and where the remaining $N$ simple poles may be considered to be time variables \cite{Okamoto1981, Garnier}. When $N=1$,  the Garnier system is equivalent to the sixth Painlev\'e equation \cite{Fuchs1, Fuchs2}. We recently presented four distinct classes of discrete integrable systems that may be considered to be discrete analogues of the Garnier system \cite{Ormerod2016}, three of which were new and one of which coincided with the $q$-Garnier system of Sakai \cite{Sakai:Garnier}. The aim of this work is to present an elliptic analogue of the Garnier system.

Our generalization is based on a discrete analogue of isomonodromy for linear systems of difference equations \cite{Borodin:connection, Sakai:qP6, Ormerodlattice, Gramani:Isomonodromic, rains:isomonodromy}. These works concern linear systems of difference equations, which may written in matrix form as
\begin{align}\label{Alinear}
\sigma Y(z) = A(z)Y(z),
\end{align}
where $\sigma$ is a shift operator. The cases that have recieved the most attention have been when $\sigma = \sigma_h : f(z) \to f(z+h)$ or $\sigma = \sigma_q: f(z)\to f(qz)$ and $A(z)$ is a rational matrix. When this is the case, the notion of monodromy is based on the existence of two canonical solutions defined by series solutions that are convergent in neighborhoods of $\Re z = \pm \infty$ in the case of $\sigma_h$ and $z = 0, \infty$ in the case of $\sigma_q$ \cite{Birkhoff, Birkhoffallied, BirkhoddAdamsSum}. The matrix of connection coefficients relating these two solutions plays a role that is analogous to the monodromy matrices. A discrete isomonodromic deformation is induced by an second linear system of the form 
\begin{align}\label{Rlinear}
\tau Y(z) = R(z)Y(z),
\end{align}
where $\tau$ is some shift operator acting on some auxiliary variables. The compatibility of these two systems ($\sigma \circ \tau Y(z) = \tau \circ \sigma Y(z)$) imposes the condition
\begin{equation}\label{asymcomp}
(\sigma R(z)) A(z) = (\tau A(z)) R(z),
\end{equation}
where the entries of $A(z)$ satisfy nonlinear difference equations as functions of the auxiliary variables of $\tau$. The first discrete analogue of isomonodromy appeared in the work of Papageorgiou et al. \cite{Gramani:Isomonodromic}, where the authors derived a discrete analogue of the third Painlev\'e equation from \eqref{asymcomp}. The connection to the work of Birkhoff was made later in a study that introduced a $q$-analogue of the sixth Painlev\'e equation by Jimbo and Sakai \cite{Sakai:qP6}. It was shown that imposing \eqref{Rlinear} where $R(z)$ satisfies \eqref{asymcomp} preserves the connection matrix of a regular (Fuchsian) system of linear $q$-difference equations where $A(x)$ is rational and singular at four values of $z$ \cite{Sakai:qP6}. Discrete isomonodromic deformations shift pairs of these singular values, which we may think of as being proportional to an auxiliary time variable, say $t$, while shifting different pairs give transformations that commute with the evolution in $t$ \cite{Ormerodlattice}. 

The discrete Garnier systems presented in \cite{Ormerod2016} arose from systems of linear difference equations of the form \eqref{Alinear} and \eqref{Rlinear}, where $A(x)$ is a rational $2\times 2$ matrix with an even number of singular points. An important property of two of the systems introduced was that the solutions satisfied an additional symmetry property, which we can generalize (see \cite{Ormerod2016a}) to systems of the form \eqref{Alinear} satifying the additional property 
\[
Y(z) = Y(\eta -z),
\] 
in the case $\sigma = \sigma_h$ and 
\begin{align}\label{symmetry}
Y(z) = Y\left( \dfrac{\eta}{z} \right),
\end{align}
in the case in which  $\sigma = \sigma_q$. This implies that the matrix $A(z)$ necessarily possesses some additional structure, as does any resulting deformation of the form \eqref{Rlinear}. We call any Lax pair, \eqref{Alinear} and \eqref{Rlinear}, with this additional property a {\em symmetric Lax pair} \cite{Ormerod2016a}. 

In this study, the elliptic Garnier system arises from a linear system of the form \eqref{Alinear} with the property \eqref{symmetry} where $\sigma=\sigma_q$ where the entries $A(z)$ are {\em explicit meromorphic theta functions} as opposed to rational functions. In this way, the coefficients of the system of difference equations have an interpretation as holomorphic sections of a line bundle on an elliptic curve. This work serves as a particular example of a linear system of difference equations on an elliptic curve, the moduli spaces of which were studied in a previous paper \cite{Rains2013}.  

In contrast with previous discrete Garnier systems, the matrix $A(z)$ in \eqref{Alinear} is singular at $4m+12$ points, which come in pairs due to \eqref{symmetry}. These pairs play the role of the auxiliary time variables. The elliptic Garnier system is defined to be the system of discrete isomonodromic deformations, which are described in terms of two fundamental types of involutions, the compositions of which are of infinite order. In this way, the fundamental involutions we present generate copies of the infinite dihedral group, which is a more natural setting for describing the difference equations arising as discrete isomonodromic deformations of symmetric linear systems of difference equations. 

In the case in which $m=1$, we obtain a system of difference equations on an elliptic curve whose relevant 2-dimensional moduli space is a rational surface which may be identified through \cite{Rains2013} with the surface of initial conditions for the elliptic Painlev\'e equation \cite{Sakai:Rational}. This means we are able provide a new and explicit Lax pair for the elliptic Painlev\'e equation \cite{NoumiYamada:ellE8Lax, rains:isomonodromy}. In this way, our system generalizes the elliptic Painlev\'e equation in the same way that the Garnier system and discrete Garnier systems generalize the sixth Painlev\'e equation and the discrete analogues of the sixth Painlev\'e equation respectively. 

This article is organized as follows: In \S \ref{Background} we review the basic background, including a brief recapitulation of the notation used for elliptic functions and some modern notion of what it means to be a discrete isomonodromic deformation in this setting. In \S \ref{linearprob} we give explicit representations of the entries of the associated linear problem for the elliptic Garnier system. In \S \ref{isomonodromic} we outline the two canonical types of transformations that generate the elliptic Garnier system. Lastly, in \S \ref{ellipticPE8} we demonstrate how specializing to the case $m=1$ gives the elliptic Painlev\'e equation.

\section{Background}\label{Background}

We fix a $q \in \mathbb{C}^{*}$ such that $|q| < 1$. We denote the $q$-Pochhammer symbol by
\begin{align}
(z;q)_\infty &= \prod_{k=0}^{\infty} (1-q^{k} z),\\
(z;q)_k &= \dfrac{(z;q)_{\infty}}{(q^kz;q)_{\infty}},
\end{align}
for $k \in \mathbb{Z}$ and
\[
(z_1,z_2,\ldots,z_n;q)_{k} = \prod_{i=1}^n (z_i;q)_k.
\] 
The building block for our associated linear problem is the $q$-theta function (see \cite{GasperRahman}), given by
\begin{align}
\theta_q(z) &= (z,q/z;q)_{\infty},\\
\theta_q(z_1,\ldots,z_n) &= \prod_{k=1}^{n} \theta_q(z_k),
\end{align}
which satisfies the relations
\begin{equation}\label{thetaidentity}
\theta_q(z) = - z \theta_q(qz) = -z \theta_q\left( \dfrac{1}{z} \right).
\end{equation}
The Jacobi triple product formula can be stated in our notation as
\begin{equation}
\theta_q(z) = \dfrac{1}{(q;q)_{\infty}} \sum_{n= -\infty}^{\infty} (-1)^{n}q^{{n \choose 2}} z^n.
\end{equation}
Many identities may be derived from the properties of theta functions as holomorphic sections of an line bundle. For example, the addition law may be derived from the fact that the space of theta functions with two zeros and fixed multiplier form a 2-dimensional vector space. Working out the precise linear dependence involves direct substitution. One of the many equivalent forms is given by
\begin{equation}\label{addition}
a \theta_q \left(\frac{z}{a},a z,\frac{b}{c},b c\right)-b \theta_q \left(\frac{a}{c}, a c,\frac{z}{b},b z\right)+b \theta_q \left(\frac{a}{b},a b,\frac{z}{c},c z\right) = 0.
\end{equation}
Such relations mean that there are multiple ways of expressing the same function. Unlike the field of rational functions, there is no canonical form for a given theta function, hence, we cannot guarantee that any particular presentation is the simplest possible presentation.

\begin{lem}\label{divisibility}
If $f(z)$ is a holomorphic theta function, satisfying $f(pz) = Cz^{-k} f(z)$, and $f(x) = 0$, then $f(z)/\theta_p(z/x)$ is also a holomorphic theta function.
\end{lem}

This lemma allows gives us an easy formulation of divisibility, allowing us to factor theta functions without relying on numerous identities or relying on a particular choice of basis. One last function required is the elliptic Gamma function,
\begin{equation}\label{Gamma}
\Gamma_{p,q}(z) = \prod_{i,j \geq 0} \dfrac{1-p^{i+1}q^{j+1}/z}{1-p^iq^jz},
\end{equation}
which satisfies the relations
\begin{align}
\Gamma_{p,q}(pz) &= \theta_q(z) \Gamma_{p,q}(z),\\
\Gamma_{p,q}(qz) &= \theta_p(z) \Gamma_{p,q}(z),\\
\Gamma_{p,q}(pq/z) &= \Gamma_{p,q}(z)^{-1}.
\end{align}

In relation to linear systems parameterized in terms of other elliptic functions \cite{Mumford1983a, Mumford1983}, we have that the four Jacobi theta functions, $\vartheta_1(z;\tau),\ldots, \vartheta_4(z;\tau)$ may be expressed in terms of $\theta_q(x)$. By letting $q = \exp(i \pi \tau)$, we have the relations
\begin{align*}
\vartheta_1(z;\tau) &= i q^{1/4}(q^2;q^2)_{\infty} e^{-\pi i z} \theta_{q^2}(e^{2\pi i z}),\\
\vartheta_2(z;\tau) &= \vartheta_1\left(z + \dfrac{1}{2};\tau\right),\\
\vartheta_3(z;\tau) &= (q^2;q^2)_{\infty} \theta_{q^2}(-qe^{2\pi i z}),\\
\vartheta_4(z;\tau) &= \vartheta_3\left(z+\dfrac{1}{2}; \tau \right).
\end{align*}
One could equally use the Jacobi elliptic theta function, $\vartheta = \vartheta_3$, to parameterize the global sections of an elliptic curve. Note that the quasi-periodicity is now expressed in terms of the functional equation
\[
\vartheta(z) = \vartheta(z+1) = - e^{-i \pi z}\vartheta(z + 2\tau). 
\]
In this way, a coherent theory of linear systems may be expressed either in terms of Jacobi elliptic functions, $\vartheta$, or $q$-theta functions with associated difference operators given by $\sigma_h$ and $\sigma_q$ respectively. An analytic approach to linear systems of difference equations involving a parameterization in terms of Jacobi theta functions was presented by Krichever \cite{Krichever2004}. We remark that the conditions imposed in \cite{Krichever2004} are not satisfied by the system we present in the next section. The Jacobi theta functions may also be used to express Jacobi's elliptic functions, $\mathrm{sn}$, $\mathrm{cn}$ and $\mathrm{dn}$ and the second derivative of $\log \vartheta$ is also the Weierstass $\wp$-function (plus a constant), which may also be used in conjunction with other related functions such as $\zeta$ and $\sigma$ \cite{Nijhoff2016}.

We now turn directly to the theory of linear systems of the form \eqref{Alinear} with the additional symmetry condition, \eqref{symmetry}. If we assume \eqref{symmetry}, then we have two equivalent ways of calculating $Y(qx)$ in terms of $Y(\eta/x)$ given by 
\[
Y(qz) = Y\left( \dfrac{\eta}{qz} \right) = A\left(\dfrac{\eta}{q z}\right)^{-1}Y\left(\dfrac{1}{z}\right) = A(z)Y\left(\dfrac{\eta}{z}\right).
\]
For this evolution to be consistent, we require that $A(z)$ satisfies the relation
\begin{equation}\label{Asym}
A(z) A\left( \dfrac{\eta}{qz} \right) = I,
\end{equation}
when $A(z)$ is not singular. The following lemma is useful in characterizing the space of matrices, $A(z)$, satisfying this property.

\begin{lem}\label{lem:existencegenatu}
Let $\mathbb{L}/\mathbb{K}$ be a quadratic field extension and $A \in \mathrm{GL}_n(\mathbb{L})$ be a matrix such that $\bar{A}A = I$, where $\bar{A}$ is the conjugation of $A$ in $\mathbb{L}$ over $\mathbb{K}$. Then there exists a matrix $B \in \mathrm{GL}_n(\mathbb{L})$ such that $A = \bar{B}B^{-1}$ and $B$ is unique up to right-multiplication by $\mathrm{GL}_n(\mathbb{K})$.
\end{lem}

This is the same specialization of Hilbert's Theorem 90 that we used in \cite{Ormerod2016} applied to the field of elliptic functions with the same periods. If the entries of $A(x)$ are in the field of elliptic functions, we apply Lemma \ref{lem:existencegenatu} where $\bar{A}(x) = A(\eta/qx)$ is the conjugation in this field to show that there exists a $B(z)$ with entries in the field of elliptic functions with the same period such that
\begin{equation}\label{product}
A(z) = B\left(\dfrac{\eta}{qz} \right)^{-1} B(z).
\end{equation}
We could consider the moduli space of symmetric matrices, $A(z)$, up to constant gauge transformations, i.e., a constant guage transformation corresponds to conjugation of $A(z)$ by constant matrices, which in turn corresponds to multiplication of $B(z)$ on the right by some matrix. In light of \eqref{product}, $B(z)$ itself is only defined up multiplication on the left by a constant matrix, hence, the moduli space of symmetric matrices, $A(z)$, up to constant guage transformations is related to the moduli space of $B(z)$ matrices up to multiplication on the left and right by constant matrices.

In order to consider our Garnier system the result of a discrete isomonodromic deformation, we require that the transformations preserve some sort of structure. The analogous notion of monodromy for systems of difference and $q$-difference equations is given by a connection matrix \cite{Birkhoff, Birkhoffallied}. Given a regular system of difference equations and $q$-difference equations of the form of \eqref{Alinear}, there exists two fundamental solutions defined in terms of series solutions that are convergent in two different regions of the complex plane. The connection matrix, which is usually denoted $P(z)$, defines the relation between these fundamental solutions in the same way that the monodromy matrices define the relation between two solutions, one of which is the solution obtained by integrating around a singularity. 

A modern interpretation of the ideas of monodromy lies in the Galois theory of difference equations \cite{Etingof1995, vanderPutSinger, Sauloy} which is closely related to differential Galois theory \cite{VanderPut2003}. The relevant difference field is obtained by adjoining either the formal symbolic solutions satisfying \eqref{Alinear} or the meromorphic solutions of \eqref{Alinear}. The group of automorphisms  preserving the difference ring structure is called the difference Galois group \cite{vanderPutSinger}. When the two fundamental solutions exist, we have two different, yet isomorphic, difference fields. The connection matrix defines an isomorphism between these two difference fields, hence, we obtain automorphisms of the difference field by considering elements of the form $P(u)^{-1} P(v)$ when defined \cite{Etingof1995}. The meromorphic functions solving both \eqref{Alinear} and \eqref{symmetry} arise as a corollary of the following theorem.

\begin{thm}[Thereom 3 of \cite{Praagman:Solutions}]\label{Praagman}
Let $G$ be a group of automorphisms of $\mathbb{P}^1$, $L$ is the limit set of $G$ and $U$ a component of $\mathbb{P}^1\setminus L$ such that $G(U) = U$. If there is a map, $G \to \mathrm{GL}_m(\mathcal{M}_U)$, $g \to A_g(x)$ satisfying 
\[
A_{gh}(x) = A_g(h(z)) A_h(z), 
\]
then the system of equations
\[
Y(\gamma(z)) = A_{\gamma}(z) Y(z) ,\hspace{2cm} \gamma \in G,
\]
possesses a meromorphic solution. 
\end{thm}

In the context of $q$-theta functions, the relevant group, $G$, is given by
\begin{align*}
G = \left\langle \tau_1, \tau_2 |  \tau_1(x) = \dfrac{\eta}{qx}, \tau_2(x) = \dfrac{\eta}{x} \right\rangle.
\end{align*}
If we let $A_{\tau_2} = I$ in each case and $A_{\tau_1}(x)$ be, $A(x)^{-1}$, we recover \eqref{Alinear} and \eqref{symmetry}. This is the same logic used in \cite{Ormerod2016} applied to a characteristically different class of linear problem. 

The Tannakian structure of the category of difference modules gives us an intrinsic definition of the Galois group. The consequence is that isomorphic difference modules have isomorphic Galois groups. Since isomonorphisms of difference modules are defined by transformations of the form \eqref{Rlinear} we may draw the following conclusion. 

\begin{cor}\label{corintegrability}
Two systems, $\sigma Y(x) = A(x)Y(x)$ and $\sigma \tilde{Y}(x) = \tilde{A}(x)\tilde{Y}(x)$,  related by \eqref{Rlinear} defines a transformation that preserves the Galois group.
\end{cor}

This defines an appropriate notion of monodromy that exists regardless of the existence of any connection matrix. It is in this sense that the transformations we may consider are discrete isomonodromic deformations. 

\begin{rem}
A more precise formulation of isomonodromy specific to $p$-theta $q$-difference equations may be found in \cite{rains:isomonodromy} which associates to every $p$-theta $q$-difference equation a $q$-theta $p$-difference equation. It is shown that the constructed $q$-theta $p$-difference equation is invariant under transformations of the form \eqref{Rlinear} satisfying \eqref{asymcomp}. Moreover, the correspondence between the system of $p$-theta $q$-difference equations and the system of $q$-theta $p$-difference equations established in \cite{rains:isomonodromy} is invertible by the same construction. This may effectively be used to establish an analogous Riemann-Hilbert type correspondence for systems of elliptic difference equations \cite{rains:isomonodromy}. 
\end{rem}

\section{The associated linear problem}\label{linearprob}

Our aim is to present a matrix $A(z)$ whose entries lie in the space of theta functions with some fixed multiplier. We may easily relate such matrices to those with coefficients in the field of elliptic functions with the same period via gauge transformations by elliptic Gamma functions. In this way, we relate $A(z)$ to a matrix $B(z)$ via \eqref{product} whose entries are in the space of holomorphic theta functions. The conditions we impose on the entries of $B(z)$ are as follows:
\begin{align}
\tag{$A_1$}\label{C1}&\begin{array}{p{14cm}}The determinant of $B(z)$ vanishes at points $u_0$, \ldots, $u_{2m+5}$ and is nonzero.\end{array}\\
\tag{$A_2$}\label{C2}&\begin{array}{p{14cm}}The kernel of $B(u_k)$ is $\langle ( x_k, y_k )\rangle$ for $0 \leq k \leq 2m+2$.\end{array}\\
\tag{$A_3$}\label{C3}&\begin{array}{p{14cm}}The images of $B(z)$ at $z= u_{2m+3}$, $u_{2m+ 4}$ and $u_{2m+5}$ are $\langle (1,1)\rangle$, $\langle (0,1)\rangle$ and $\langle (1,0)\rangle$ respectively.
\end{array}
\end{align}
So long as $u_0, \ldots, u_{2m+5}$ are distinct, and since the entries of $B(z)$ lie in a $(m+3)$-dimensional vector space, these conditions are sufficient to specify $B(z)$ up to multiplication by some scalar matrix, i.e., this is sufficient to uniquely define $A(z)$. To succinctly specify $B(z)$, we first define the following notation; we fix an $m$, then if $S \subset \mathbb{Z}_{2m+3} = \{ 0,1,\ldots, 2m+2\}$, then we define  
\begin{align*}
&x_S = \prod_{i \in S} x_i && u_S= \prod_{i \in S} u_i,  && y_S = \prod_{i \in S} y_i,\\
&x_{\bar{S}} = \prod_{\stackrel{i \in \mathbb{Z}_{2m+3}}{i \notin S}} x_i && u_{\bar{S}}= \prod_{\stackrel{i \in \mathbb{Z}_{2m+3}}{i \notin S}} u_i,  && y_{\bar{S}}=   \prod_{\stackrel{i \in \mathbb{Z}_{2m+3}}{i \notin S}} y_i.
\end{align*}
We claim that we may satisfy conditions \eqref{C1}, \eqref{C2} and \eqref{C3} by letting
\begin{align}\label{Bentries}
B(z) = \begin{pmatrix} B_{11}(z) & B_{12}(z) \\ B_{21}(z) & B_{22}(z) \end{pmatrix}
\end{align}
where the entries may be specified by
\begin{subequations}\label{BGarnier}
\begin{align}
B_{11}(z) =
\dfrac{\theta_p\left(\frac{z}{u_{2m+4}}\right)}
     {\theta_p\left(\frac{u_{2m+3}}{u_{2m+4}}\right)}
&\sum_{\stackrel{S\subset \mathbb{Z}_{2m+3}}{|S|=m+1}}
\dfrac{x_{S}y_{\bar{S}}}{u_{\bar{S}}u_{2m+4}}
\theta_p\left(\frac{u_{S}u_{2m+4}z}{L},\frac{u_{\bar{S}}u_{2m+4}}{L}\right) \\
                            &\times \prod_{i\in S} \theta_p\left(\frac{z}{u_i}\right) \prod_{\stackrel{i\in S}{j\notin S}} u_j^{-1}\theta_p\left(\dfrac{u_i}{u_j}\right)^{-1},\nonumber
\end{align}
\begin{align}
B_{12}(z) =
\frac{\theta_p\left(\frac{z}{u_{2m+4}}\right)}
     {\theta_p\left(\frac{u_{2m+3}}{u_{2m+4}}\right)}
&\sum_{\stackrel{S\subset \mathbb{Z}_{2m+3}}{|S|=m+2}} \dfrac{(-1)^mx_Sy_{\bar{S}}}{u_{S}u_{2m+4}} 
\theta_p\left(\frac{u_{2m+4}u_S}{L}, \frac{u_{2m+4}zu_{\bar{S}}}{L}\right)
\\
&\times \prod_{i\notin S} \theta_p\left(\frac{z}{u_i}\right) \prod_{\stackrel{i\in S}{j\notin S}} u_j^{-1}\theta_p\left(\dfrac{u_i}{u_j}\right)^{-1},\nonumber
\end{align}
\begin{align}
B_{21}(z) =
\dfrac{\theta_p\left(\frac{z}{u_{2m+5}}\right)}
     {\theta_p\left(\frac{u_{2m+3}}{u_{2m+5}}\right)}
&\sum_{\stackrel{S\subset \mathbb{Z}_{2m+3}}{|S|=m+1}} \dfrac{x_Sy_{\bar{S}}}{u_{\bar{S}}u_{2m+5}}
\theta_p\left(\frac{u_{2m+5}u_Sz}{L},
                           \frac{u_{2m+5}u_{\bar{S}}}{L}\right)\\
&\times \prod_{i\in S} \theta_p\left(\frac{z}{u_i}\right)
\prod_{\stackrel{i\in S}{j\notin S}}  u_j^{-1}\theta_p\left(\dfrac{u_i}{u_j}\right)^{-1},\nonumber
\end{align}
\begin{align}
B_{22}(z) =
\dfrac{\theta_p\left(\frac{z}{u_{2m+5}}\right)}
     {\theta_p\left(\frac{u_{2m+3}}{u_{2m+5}}\right)}
& \sum_{\stackrel{S\subset \mathbb{Z}_{2m+3}}{|S|=m+2}} \dfrac{(-1)^mx_Sy_{\bar{S}}}{u_{S}u_{2m+5}}\theta_p\left(\frac{u_{2m+5}u_S}{L},
                            \frac{u_{2m+5}u_{\bar{S}} z}{L}\right)
\\
&\times \prod_{i\notin S} \theta_p\left(\frac{z}{u_i}\right)
\prod_{\stackrel{i\in S}{j\notin S}}  u_j^{-1}\theta_p\left(\dfrac{u_i}{u_j} \right)^{-1},\nonumber
\end{align}
\end{subequations}
where $L$ satisfies
\begin{equation}\label{Ldef}
L^2 = \prod_{j=0}^{2m+5} u_j.
\end{equation}
Firstly, each of these entries have the property $f(pz) = L z^{-(m+3)}f(z)$, which defines the fixed multiplier. Secondly, while the entries of $B(z)$ depend the specific choices of homogeneous coordinates, $(x_k,y_k) \in \mathbb{C}^2$, the entries are homogeneous, and thus $A(z)$ depends only on the corresponding points in $\mathbb{P}^1$.

\begin{thm}
The matrix $B(z)$ has properties \eqref{C1}, \eqref{C2} and \eqref{C3}, furthermore, $A(z)$ is uniquely determined by \eqref{C1}, \eqref{C2} and \eqref{C3}.
\end{thm}

\begin{proof}
To show \eqref{C3} and \eqref{C1} for $z= u_{2m+4}$, $z = u_{2m+5}$, we see that 
\[
B_{11}(u_{2m+4}) =B_{11}(u_{2m+4}) =  B_{21}(u_{2m+5}) = B_{22}(u_{2m+5}) = 0,
\]
showing that $\det B(u_{2m+4}) = \det B(z_{2m+5}) = 0$ and that 
\begin{align*}
\mathrm{Im}B(u_{2m+4}) = \left\langle (0,1) \right\rangle, \hspace{.5cm} \mathrm{Im}B(u_{2m+5}) = \left\langle (1,0) \right\rangle.
\end{align*}
To show \eqref{C3} and \eqref{C1} for $z = u_{2m+3}$, we evaluate $B_{11}(z) - B_{21}(z)$ at $z = u_{2m+3}$ to see
\begin{align*}
B_{11}&(u_{2m+3}) - B_{21}(u_{2m+3}) = \sum_{\stackrel{S\subset \mathbb{Z}_{2m+3}}{|S|=m+1}} \dfrac{x_{S}y_{\bar{S}}}{u_{\bar{S}}} \left[u_{2m+4}^{-1} \theta_p\left( \dfrac{u_Su_{2m+3}u_{2m+4}}{L}, \dfrac{u_{\bar{S}}u_{2m+4}}{L}\right)\right. \\
&\left. -  u_{2m+5}^{-1}\theta_p\left( \dfrac{u_Su_{2m+3}u_{2m+5}}{L}, \dfrac{u_{\bar{S}}u_{2m+5}}{L}\right) \right]  \prod_{i\in S} \theta_p\left(\frac{z}{u_i}\right) \prod_{\stackrel{i\in S}{j\notin S}} u_j^{-1}\theta_p\left(\dfrac{u_i}{u_j}\right)^{-1},
\end{align*}
where the term inside the square brackets is $0$, which can be shown using \eqref{thetaidentity} and \eqref{Ldef}. We see that $B_{12}(u_{2m+3}) = B_{22}(u_{2m+3})$ for similar reasons. This shows that 
\[
B(u_{2m+3}) \begin{pmatrix} 1 \\ 0 \end{pmatrix} \in \langle \begin{pmatrix} 1 \\ 1 \end{pmatrix} \rangle, \hspace{1cm} B(u_{2m+3}) \begin{pmatrix} 0 \\ 1 \end{pmatrix} \in \langle \begin{pmatrix} 1 \\ 1 \end{pmatrix} \rangle,
\] 
which shows that $\mathrm{Im}B(u_{2m+3}) = \langle (1,1) \rangle$ and hence $\det B(u_{2m+3}) = 0$. 

To show \eqref{C2}, and complete the proof of property \eqref{C1}, we let $k \in \{ 0, \ldots, 2m+2\}$, for which we claim that 
\begin{equation}\label{keruk}
x_k B_{11}(u_k) + y_k B_{12}(u_k) = x_k B_{21}(u_k) + y_k B_{22}(u_k) = 0.
\end{equation}
The summands of $B_{11}$ in which $k \in S$ vanish at $z= u_k$, hence, we write $B_{11}$ as a sum over subsets of $\mathbb{Z}_{2m+3}\setminus\{k\}$ of size $m+1$. Similarly, the summands of $B_{12}$ vanish if $k \notin S$, but given a subset of $\mathbb{Z}_{2m+3}$ of size $m+1$ in which $k\notin S$, taking the union with $\{k\}$ defines a one-to-one correspondence with the subsets of $\mathbb{Z}_{2m+3}$ of size $m+2$ that include $k$. We may use this to write the left hand side of \eqref{keruk} as a single sum over subsets of $\mathbb{Z}\setminus\{k\}$ of size $m+1$. In this way, we may write the left hand side of \eqref{keruk} as
\begin{align*}
&x_k B_{11}(u_k) + y_k B_{12}(u_k) = \dfrac{\theta_p\left(\frac{u_k}{u_{2m+4}}\right)}{\theta_p\left(\frac{u_{2m+3}}{u_{2m+4}}\right)} \sum_{\stackrel{S\subset \mathbb{Z}_{2m+3} }{\stackrel{|S|=m+1}{k\notin S}}} \dfrac{x_k x_S y_{\bar{S}}}{u_{2m+4}} \theta_p\left(\frac{u_{S}u_{2m+4}u_k}{L}, \frac{u_{2m+4}u_{\bar{S}}}{L}\right)\\
& \left[\dfrac{1}{u_{\bar{S}}}\prod_{i\in S} \theta_p\left(\dfrac{u_k}{u_i}\right)  +\dfrac{(-1)^m}{u_S}  \prod_{i\notin S} \theta_p\left(\dfrac{u_k}{u_i}\right) \left(\prod_{j \notin S} u_j^{-1} \theta_p\left(\dfrac{u_k}{u_j} \right) \right)^{-1} \prod_{i \in S} u_k^{-1} \theta_p\left(\dfrac{u_i}{u_k} \right) \right] \times \prod_{\stackrel{i \in S}{j \notin S}} u_j^{-1} \theta_p \left( \dfrac{u_i}{u_j} \right),
\end{align*}
where the term inside the square brackets vanishes due to \eqref{thetaidentity}.

To show that the overall determinant is non-zero we specialize the kernels in the following way, if we let $P = \{ 1,\ldots, m+1\}$ and $Q = \{m+2,\ldots, 2m+2\}$
\[
(x_k,y_k) = \left\{ \begin{array}{c p{5cm}}
(1,1) & if  $k = 0$, \\
(1,0) & if $k \in P$,\\ 
(0,1) & if $k \in Q$. 
\end{array}\right.
\]
Specializing the kernels of $B(z)$ in this way means each entry of $B(z)$ has just one non-zero summand. This simplifies $B(z)$ to the following 
\begin{align*}
B(z) =& u_0^{-1}\prod_{\stackrel{i \in P}{j \in Q}} u_j^{-1} \theta_p\left(\frac{u_i}{u_j}\right)^{-1}
 \begin{pmatrix}
\frac{\theta_p\left(\frac{z}{u_{2m+4}} \right) }{u_{2m+4} \theta_p\left( \frac{u_{2m+3}}{u_{2m+4}}\right)} & 0 \\
0 &\frac{\theta_p\left(\frac{z}{u_{2m+5}} \right) }{u_{2m+5} \theta_p\left( \frac{u_{2m+3}}{u_{2m+5}}\right)} 
 \end{pmatrix} \\
& \begin{pmatrix}
\theta_p \left(\dfrac{u_{2m+4}zu_P}{L},\dfrac{u_{2m+4}u_0u_Q}{L}\right) &
\theta_p \left(\dfrac{u_{2m+4}u_0u_P}{L},\dfrac{zu_{2m+4}zu_Q}{L}\right)\\
-\theta_p \left(\dfrac{u_{2m+5}zu_P}{L},\dfrac{u_{2m+5}u_0u_Q}{L}\right) &
-\theta_p \left(\dfrac{u_{2m+5}u_0u_P}{L},\dfrac{u_{2m+5}zu_Q}{L}\right)
\end{pmatrix}\\
& \prod_{i\in I_1} \begin{pmatrix} \dfrac{\theta_p\left( \frac{z}{u_i} \right)}{u_0  \theta_p\left( \frac{u_i}{u_0} \right)} & 0 \\ 
0 & -\dfrac{1}{u_i} \end{pmatrix}
 \prod_{j\in I_2} \begin{pmatrix} \dfrac{1}{u_j} & 0 \\ 
0 & \dfrac{\theta_p\left( \frac{z}{u_j} \right)}{u_j  \theta_p\left( \frac{u_0}{u_j} \right)} \end{pmatrix}.
\end{align*} 
It is clear from this factorization that the determinant is non-zero provided
\begin{align*}
\theta_p \left(\dfrac{u_{2m+4}zu_P}{L},\dfrac{u_{2m+4}u_0u_Q}{L},\dfrac{u_{2m+5}u_0u_P}{L} ,\dfrac{u_{2m+5}zu_Q}{L}\right) - \\
\theta_p \left(\dfrac{u_{2m+4}u_0u_P}{L},\dfrac{zu_{2m+4}zu_Q}{L} \dfrac{u_{2m+5}zu_P}{L},\dfrac{u_{2m+5}u_0u_Q}{L}\right)\neq 0,
\end{align*}
which is generically non-zero. Furthermore, by degree considerations, these are the only zeroes on $\mathbb{C}^*/\langle p \rangle$.

For the second part of the proof, we suppsoe $\tilde{B}(z)$ is a matrix satisfying \eqref{C1}, \eqref{C2} and \eqref{C3}, then let $C(z) := B(z)\tilde{B}(z)^{-1}$. It should be clear the matrix $C(z)$ has elliptic entries, has constant determinant due to \eqref{C1} and has (simple) poles only at $z= u_{2m+3}$, $u_{2m+4}$ and $u_{2m+5}$. Both the image and kernel of the residue of $C(z)$ at $z = u_{2m+3}$ are $\langle (1,1)\rangle$. Similarly, the residues of $C(z)$ at $z= u_{2m+4}$ and $z_{2m+5}$ have images and kernels in $\langle (0,1)\rangle$ and $\langle (1,0)\rangle$. These conditions are sufficient to show that these three residues are proportional to
\[
\begin{pmatrix} 1 & -1 \\ 1 & -1 \end{pmatrix}, \hspace{.5cm} \begin{pmatrix} 0 & 0 \\ 1 & 0 \end{pmatrix}, \hspace{.5cm} \begin{pmatrix} 0 & 1 \\ 0 & 0 \end{pmatrix}
\]
respectively. Since the diagonal entries have only one nonzero residue, they must be constant, and thus $C(z)$ is actually holomorphic at $z=u_{2m+3}$. But then the off-diagonal entries have only one nonzero residue, making $C(z)$ holomorphic at $u_{2m+4}$ and $u_{2m+5}$ as well.  In particular, it follows that $C(z)$ is constant. Since $C(z)\tilde{B}(z)=B(z)$, $C(z)$ must preserve the images of $\tilde{B}(z)$ at $z=u_{2m+3}$, $u_{2m+4}$ and $u_{2m+5}$, and thus must be a scalar matrix.
\end{proof}

The values of the spectral parameters in which $\det B(z)$ vanishes, in this case $u_0,\ldots, u_{2m+5}$, are considered as independent auxiliary parameters that change under the action of discrete isomonodromic deformations \cite{Borodin:connection, Sakai:qP6, Ormerodlattice, Sakai:Garnier}.  While simplifying the determinant of $B(z)$ using the relations such as \eqref{thetaidentity} and \eqref{addition} for general $m$ might be possible, we circumvent the use of these identities by using Lemma \ref{divisibility}. By \eqref{C1} and Lemma \ref{divisibility} we have that
\[
\det (B(z)) = C \prod_{k \in \mathbb{Z}_{2m+6}} \theta_p\left(\dfrac{z}{u_k} \right),
\]
for some constant $C$. We seek isomonodromic deformations that change the values of the $u_i$ in the next section. 

It is convenient to endow this linear system with the natural action of the symmetric group on a finite set of $2m+6$ symbols, $S_{2m+6}$, which acts on the $u_i$ by permutation. We denote the generators by
\[
S_{2m+6} = \langle s_i \cdot u_i \leftrightarrow u_{i+1} | i = 0,\ldots, 2m+4 \rangle
\]
whose nontrivial effect on the kernels of $B(z)$ is specified by the following lemma.

\begin{lem}
The matrix $A(x)$ is invariant under the action of the symmetric group, $S_{2m+6}$, where the action of the generators, $s_i$, on $(x_k:y_k)$ is trivial for $i = 2m+3$ and $i = 2m+4$, specified by 
\begin{align*}
s_i \cdot (x_k:y_k) = \left\{ \begin{array}{c p{5cm}} (x_k:y_k) & for $k \neq i, i+1$\\ 
(x_i:y_i) & if $k=i+1$,\\
(x_{i+1}:y_{i+1}) & if $k = i$, 
\end{array}\right.
\end{align*}
for $i = 0, \ldots, 2m+1$, and for $i = 2m+2$ is specified by
\begin{align}
s_{2m+2} \cdot (x_{2m+2}:y_{2m+2}) = (B_{12}(u_{2m+3}) : - B_{11}(u_{2m+3})).
\end{align}
\end{lem}

\begin{proof}
If we permute $u_{2m+3}$, $u_{2m+4}$ and $u_{2m+5}$ so that the image subspaces are not $\langle (1,1) \rangle$, $\langle (1,0) \rangle$ and $\langle (0,1) \rangle$ respectively, we have a matrix that sends the images $B(u_{2m+3})$, $B(u_{2m+4})$ and $B(u_{2m+5})$ back to the subspaces $\langle (1,1)\rangle$, $\langle (1,0)\rangle $ and $\langle (0,1)\rangle$ respectively. Since $A(z)$ is invariant under left multiplication of $B(z)$ by a constant matrix, this action is considered trivial with respect to $A(z)$. Similarly, if we swap $u_i$ with $u_{i+1}$ for $i = 0,\ldots, 2m+1$, we need only swap the relevant kernels. The only nontrivial action occurs if we swap a $u_i$ associated with a kernel with one that is associated with an image, i.e., this only happens for the generator $s_{2m+2}$. When we swap $u_{2m+2}$ with $u_{2m+3}$, then $(x_{2m+2}:y_{2m+2})$ is sent to the kernel of $B(u_{2m+3})$, which is as shown above, whereas we may send the image of $u_{2m+2}$ to $\langle (1,1) \rangle$ while fixing $\langle (0,1) \rangle$ and $\langle (1,0) \rangle$ by a constant (diagonal) matrix multiplication on the left which does not change $A(z)$.
\end{proof}

\section{The discrete isomonodromic deformations} \label{isomonodromic}

Having established the properties of the associated linear problem, we need to describe the group of isomonodromic deformations. Following the previous cases of discrete Garnier systems, what is required is that the translations that define the Garnier system are specified by taking two roots of the determinant, $z = u_i$ and $z= u_j$, and shifting them in some natural manner \cite{Sakai:Garnier, Ormerod2016}. As a $q$-difference equation in the variable $z$, where $z$ appears in the arguments of $p$-theta functions, we expect to move the variables $u_i$ and $u_j$ by multiplication by $q$, which has the interpretation in terms of the addition law on an elliptic curve.

We wish to specify two canonical involutive forms of discrete isomonodromic deformations as actions on $B$; those induced by an action on the left, and those induced by an action on the right. This is the natural setting for symmetric systems of difference equations as discrete isomonodromic deformations of the form \eqref{Rlinear}, in which the transformation relating $A(z)$ and $\tilde{A}(z)$, given by \eqref{asymcomp}, is equivalent to multiplying $B(z)$ by $R(z)^{-1}$ on the right, while $B(z)$ itself is only defined up to multiplication on the left. For the following discussion, we fix $i$ and $j$ and denote the discrete isomonodromic action induced by left multiplication by a matrix $R_l(z)$ by 
\begin{subequations}
\begin{align}\label{left}
\lambda_l(z) (E_{i,j} B(z)) = R_l(z) B(z),
\end{align}
while we denote the discrete isomonodromic deformation induced by multiplication on the right by $R_r(z)$ by 
\begin{align}\label{right}
\lambda_r(z) (F_{i,j} B(z)) = B(z)R_r(z),
\end{align}
where $\lambda_{l}(z)$ and $\lambda_{r}(z)$ are scalar factors. These scalar factors play no role in the definition of the Galois group, hence, we consider multiplying $A(z)$ by a scalar to be a trivial action. The consistency with \eqref{product} requires that $R_l(x)$ and $R_r(x)$ have the symmetry
\begin{align}
R_l(z) &= R_{l}\left(\dfrac{\eta}{qz} \right), \\
R_r(z) &= R_{r}\left( \dfrac{\eta}{z} \right).
\end{align}
\end{subequations}
If we are taking the points $u_i$ and $u_j$, we set $\det R_l(u_i) = \det R_l(u_j) =0$ ($\det R_r(u_i) = \det R_r(u_j) = 0$) then from the symmetry, we expect $\det R_l(\eta/qu_i) = \det R_l(\eta/qu_j)= 0$ ($\det R_r(\eta/u_i) = \det R_r(\eta/u_i) = 0$). In particular, by Lemma \ref{divisibility} we have that 
\begin{align*}
\lambda_l(z)^2\det (E_{i,j} B(z)) = \det R_l(z) B(z) = \theta_p\left( \dfrac{z}{u_i},\dfrac{z}{u_j}\right)^2\theta_p\left( \dfrac{qzu_i}{\eta}, \dfrac{qzu_j}{\eta}\right)C_l(z),\\
\lambda_r(z)^2\det (F_{i,j} B(z)) = \det B(z)R_r(z) = \theta_p\left( \dfrac{z}{u_i},\dfrac{z}{u_j}\right)^2\theta_p\left( \dfrac{zu_i}{\eta}, \dfrac{zu_j}{\eta}\right)C_r(z),
\end{align*}
indicating that the action of $E_{i,j}$ and $F_{i,j}$ fixes $u_{k}$ for $k \neq i,j$ and transforms $u_i$ and $u_j$ as follows
\[
E_{i,j} \cdot u_{i,j} \to \dfrac{\eta}{qz}, \hspace{1cm} F_{i,j} \cdot u_{i,j} \to \dfrac{\eta}{z},
\]
where the action on the kernels of $B(z)$ are yet to be specified. Secondly, we have that
\begin{align*}
\lambda_{l}(u_i) = \lambda_l(u_j) = \lambda_r(u_i) = \lambda_r(u_j)= 0.
\end{align*}
It follows that we may take 
\[
\lambda_l(z) = \lambda_r(z) = \theta_p\left( \dfrac{z}{u_i},\dfrac{z}{u_j}\right).
\]
We start with the conditions that define $R_{r}(z)$. We have that the entries of $R_r(z)$ are expressible in terms of functions of the form of $\theta_p(z/a, \eta/zqa)$ so that the symmetry constraint is satisfied. Setting $z= u_i$ and $z= u_j$ in \eqref{right} tells us
\begin{align}
\tag{$R_1$}\label{R1}&\begin{array}{p{14cm}} The determinant of $R_r(z)$ vanishes at points $z=u_i$ and $z=u_j$ and is nonzero.\end{array}\\
\tag{$R_2$}\label{R2}&\begin{array}{p{14cm}} The images of $R_r(u_i)$ and $R_r(u_j)$ are contained in the kernel of $B(u_i)$ and $B(u_j)$ respectively.\end{array}
\end{align}
In order to avoid any confusion, we simply define $(x_k:y_k)$ for general $k \in \mathbb{Z}_{2m+6}$ to be any vector $(x_k,y_k)$ generating the kernel of $B(u_k)$ which is consistent with $(x_k:y_k)$ for $k \in \mathbb{Z}_{2m+3}$. If the entries are in the space of theta functions with a fixed multiplier and two zeros, these conditions only determine $R_r(z)$ up to some diagonal matrix. We impose the additional constraint that 
\begin{align}
\tag{$R_3$}\label{R3}\begin{array}{p{14cm}} There is a fixed $v$ such that $R_r(v) = I$.\end{array}
\end{align}
This was implicitly done in the $q$-difference and $h$-difference cases for $v = \infty$, however, in the case of elliptic functions, no such distinguished point exists. We just need to specify a basis for the entries of $R_{r}(z)$ and determine the coefficients in accordance with \eqref{R1}, \eqref{R2} and \eqref{R3}. We claim the following is a presentation of the required matrix:
\begin{align*}
R_r(x) = \frac{1}{x_jy_i - x_iy_j}& \left( \begin{array}{c c} 
x_j y_i \dfrac{\theta_p\left(\frac{z}{u_i},\frac{zu_i}{\eta} \right)}{\theta_p\left(\frac{v}{u_i},\frac{vu_j}{\eta} \right)} - x_i y_j \dfrac{\theta_p\left(\frac{z}{u_j},\frac{zu_j}{\eta} \right)}{\theta_p\left(\frac{v}{u_j},\frac{vu_j}{\eta} \right)}
& x_ix_j \left(\dfrac{\theta_p\left(\frac{z}{u_j},\frac{zu_j}{\eta} \right)}{\theta_p\left(\frac{v}{u_j},\frac{vu_j}{\eta} \right)}-  \dfrac{\theta_p\left(\frac{z}{u_i},\frac{zu_i}{\eta} \right)}{\theta_p\left(\frac{v}{u_i} ,\frac{vu_i}{\eta} \right)}\right)\\
y_iy_j\left(  \dfrac{\theta_p\left(\frac{z}{u_i} ,\frac{zu_i}{\eta} \right)}{\theta_p\left(\frac{v}{u_i},\frac{vu_i}{\eta} \right)} - \dfrac{\theta_p\left(\frac{z}{u_j},\frac{zu_j}{\eta} \right)}{\theta_p\left(\frac{v}{u_j}\frac{vu_j}{\eta} \right)} \right) 
& x_jy_i \dfrac{\theta_p\left(\frac{z}{u_j},\frac{zu_j}{\eta} \right)}{\theta_p\left(\frac{v}{u_j},\frac{vu_j}{\eta} \right)}-  x_iy_j\dfrac{\theta_p\left(\frac{z}{u_i},\frac{zu_i}{\eta} \right)}{\theta_p\left(\frac{v}{u_i} ,\frac{vu_i}{\eta} \right)}
\end{array} \right).
\end{align*}
It should be clear that
\[
\det R_r(z) = \dfrac{\theta_p \left(\dfrac{z}{u_i},\dfrac{z}{u_j},\dfrac{\eta}{zu_i},\dfrac{\eta}{zu_j}\right)}{\theta_p \left(\dfrac{v}{u_i},\dfrac{v}{u_j},\dfrac{\eta}{vu_i},\dfrac{\eta}{vu_j}\right)},
\]
and is easy to verify that this matrix possesses properties \eqref{R1}, \eqref{R2} and \eqref{R3} by hand.

\begin{thm} 
The relation between $F_{i,j} \cdot (x_k:y_k) = (\tilde{x_k}:\tilde{y}_k)$ and the $(x_k:y_k)$ for $k \neq i,j$ is 
\begin{align}\label{eijaction}
(\tilde{x}_k:\tilde{y}_k) =& \left( u_j x_i \left(y_j x_k-x_j y_k\right) \theta
   _p\left(\frac{u_i}{u_k},\frac{v}{u_j},\frac{\eta }{u_i u_k},\frac{\eta }{u_j v}\right) \right.\\
   &+u_i x_j
   \left(x_i y_k-y_i x_k\right) \theta
   _p\left(\frac{u_j}{u_k},\frac{v}{u_i},\frac{\eta }{u_j u_k},\frac{\eta }{u_i v}\right)\nonumber\\
&: u_j y_i \left(y_j x_k-x_j y_k\right) \theta
   _p\left(\frac{u_i}{u_k},\frac{v}{u_j},\frac{\eta }{u_i u_k},\frac{\eta }{u_j v}\right) \nonumber\\
   &\left. +u_i y_j
   \left(x_i y_k-y_i x_k\right) \theta
   _p\left(\frac{u_j}{u_k},\frac{v}{u_i},\frac{\eta }{u_j u_k},\frac{\eta }{u_i v}\right) \right), \nonumber
\end{align}
where for $k = i,j$ we have
\[
(\tilde{x}_i,\tilde{y}_i) = (x_j, y_j), \hspace{1cm} (\tilde{x}_j,\tilde{y}_j) = (x_i, y_i).
\]
\end{thm}

\begin{proof}
Given \eqref{right}, the kernel of $\tilde{B}(u_k)$ may be computed by evaluating the kernel of $B(u_k)R_r(u_k)$ which is equivalent to asking what vector is in the image of the kernel of $B(z)$, i.e., we solve
\[
 \begin{pmatrix} x_k \\ y_k \end{pmatrix}   = R_{r}(u_k) \begin{pmatrix} \tilde{x}_k \\ \tilde{y}_k \end{pmatrix},
\]
which provides a representative of $(\tilde{x}_k:\tilde{y}_k)$ directly. We also need to evaluate the kernel of $\tilde{B}(\eta/qu_i)$ and $\tilde{B}(\eta/qu_j)$, which is just the kernel of $R(u_i)$ and $R(u_j)$ by the symmetry, which is simply given by swapping the kernels of $B(u_i)$ and $B(u_j)$. The resulting matrix, $B(z) R_r(z)$ has entries divisible by $\lambda_l(z) = \theta_p(z/u_i)\theta_p(z/u_j)$, with properties \eqref{C1}, \eqref{C2} and \eqref{C3}, hence, defines the required transformation.
\end{proof}

We may now specify $R_l(z)$ using conditions equivalent to \eqref{R1}, \eqref{R2} and \eqref{R3}. By evaluating \eqref{left} at $z= u_i$ and $u_j$, we have the following conditions:
\begin{align}
\tag{$L_1$} \label{L1}&\begin{array}{p{14cm}} The determinant of $R_l(z)$ vanishes at points $z=u_i$ and $z=u_j$ and is nonzero.\end{array}\\
\tag{$L_2$} \label{L2}&\begin{array}{p{14cm}} The kernels of $R_l(u_i)$ and $R_l(u_j)$ are contained in the images of $B(u_i)$ and $B(u_j)$ respectively.\end{array}\\
\tag{$L_3$} \label{L3}&\begin{array}{p{14cm}} There is a point, $w$, such that $R_l(w) = I$.\end{array}
\end{align}
These conditions are sufficient to determine $R_{l}(z)$, whose entries may be written in terms of $\theta_p(z/u_i,zu_i/\eta)$ and $\theta_p(z/u_i,zu_i/\eta)$. We claim that the following matrix satisfies these conditions:
\begin{align}\label{RL}
&R_l(z) = \dfrac{1}{B_{11}(u_i)B_{21}(u_j) - B_{11}(u_j)B_{21}(u_i)} \begin{pmatrix} C_{11}(z) & C_{12}(z) \\ C_{21}(z) & C_{22}(z) \end{pmatrix},
\end{align}
where
\begin{align}
C_{11}(z) &=  B_{11}\left(u_i\right) B_{21}\left(u_j\right)\frac{\theta_p\left(\frac{z}{u_i},\frac{\eta }{q z u_i}\right)}{\theta_p\left(\frac{w}{u_i},\frac{\eta }{q w
   u_i}\right)}-B_{21}\left(u_i\right) B_{11}\left(u_j\right)\frac{\theta_p\left(\frac{z}{u_j},\frac{\eta }{q z u_j}\right)}{\theta _p\left(\frac{w}{u_j},\frac{\eta
   }{q w u_j}\right)},\\
C_{12}(z) &= B_{11}\left(u_i\right) B_{11}\left(u_j\right) \left(\frac{\theta _p\left(\frac{z}{u_j},\frac{\eta }{q z u_j}\right)}{\theta _p\left(\frac{w}{u_j},\frac{\eta }{q w
   u_j}\right)}-\frac{\theta _p\left(\frac{z}{u_i},\frac{\eta }{q z u_i}\right)}{\theta _p\left(\frac{w}{u_i},\frac{\eta }{q w u_i}\right)}\right),\\
C_{21}(z)&= B_{21}\left(u_i\right) B_{21}\left(u_j\right) \left(\frac{\theta_p\left(\frac{z}{u_i},\frac{\eta }{q z u_i}\right)}{\theta _p\left(\frac{w}{u_i},\frac{\eta }{q w
   u_i}\right)}-\frac{\theta _p\left(\frac{z}{u_j},\frac{\eta }{q z u_j}\right)}{\theta_p\left(\frac{w}{u_j},\frac{\eta }{q w u_j}\right)}\right), \\
C_{22}(z)&= B_{11}\left(u_i\right) B_{21}\left(u_j\right)\frac{\theta _p\left(\frac{z}{u_j},\frac{\eta }{q z u_j}\right)}{\theta _p\left(\frac{w}{u_j},\frac{\eta }{q w
   u_j}\right)}-B_{21}\left(u_i\right) B_{11}\left(u_j\right)\frac{\theta _p\left(\frac{z}{u_i},\frac{\eta }{q z u_i}\right)}{\theta _p\left(\frac{w}{u_i},\frac{\eta }{q w u_i}\right)}.
\end{align}
The determinant of this matrix is
\[
\det R_l(z) = \dfrac{\theta_p\left(\dfrac{z}{u_i},\dfrac{z}{u_j},\dfrac{\eta}{qzu_i},\dfrac{\eta}{qzu_i} \right)}{\theta_p\left(\dfrac{w}{u_i},\dfrac{w}{u_j},\dfrac{\eta}{qwu_i},\dfrac{\eta}{qwu_i} \right)},
\]
where the other properties, \eqref{L2} and \eqref{L3}, follow naturally.

\begin{thm}
The action of $E_{i,j} \cdot (x_k:y_k) = (\hat{x}_k,\hat{y}_k)$, is given by $(\hat{x}_k:\hat{y}_k) = (x_k:y_k)$ for $k \neq i,j$, while for $k = i,j$ we have
\begin{align}\label{fijaction}
(\hat{x}_i:\hat{y}_i) = \biggl(&B_{21}(u_i)B_{21}\left(\dfrac{\eta}{qu_i}\right)-B_{11}(u_i)B_{22}\left(\dfrac{\eta}{qu_i} \right): \\ 
&\left. B_{11}(u_i) B_{21}\left( \dfrac{\eta}{qu_i} \right)- B_{21}(u_i)B_{11}\left(\dfrac{\eta}{qu_i} \right)\right)  \nonumber \\
(\hat{x}_j:\hat{y}_j) = \biggl(& B_{21}(u_j)B_{21}\left(\dfrac{\eta}{qu_j}\right)-B_{11}(u_j)B_{22}\left(\dfrac{\eta}{qu_j} \right): \\ 
&\left. B_{11}(u_j) B_{21}\left( \dfrac{\eta}{qu_j} \right)- B_{21}(u_j)B_{11}\left(\dfrac{\eta}{qu_j} \right)\right). \nonumber
\end{align}
\end{thm}

\begin{proof}
It is trivial to see that the multiplication of a matrix on the left leaves the kernels of $B(z)$ at $z=x_k$ for $k\neq i,j$ unchanged, hence,  $(\hat{x}_k:\hat{y}_k) = (x_k:y_k)$ for $k \neq i,j$. For $k =i,j$, we note that $E_{i,j}B(z)$ is singular at $z = \eta/qu_i$ and $z= \eta/qu_j$, hence, we need to evaluate the kernel of $R_l(\eta/qu_i)B(\eta/qu_i)$ and $R_l(\eta/qu_j)B(\eta/qu_j)$. Since $B(\eta/qu_i)$ and $B(\eta/qu_j)$ are not singular and the kernels of $R_l(u_i)$ and $R_l(u_j)$ are $\langle (B_{11}(u_i): B_{21}(u_i)) \rangle$ and $\langle (B_{11}(u_j): B_{21}(u_j)) \rangle$, we obtain the kernel of $F_{i,j} B(z)$ at $z= \eta/qu_i$ and $z= \eta/qu_j$ by solving
\[
\begin{pmatrix} B_{11}(u_i) \\ B_{21}(u_i) \end{pmatrix} = B\left(\dfrac{\eta}{qu_i}\right) \begin{pmatrix} \hat{x}_i \\ \hat{y}_i \end{pmatrix}, \hspace{1cm} \begin{pmatrix} B_{11}(u_j) \\ B_{21}(u_j) \end{pmatrix} = B\left(\dfrac{\eta}{qu_j}\right) \begin{pmatrix} \hat{x}_j \\ \hat{y}_j \end{pmatrix},
\] 
which is equivalent to \eqref{fijaction}.
\end{proof}

It should be noted that the conditions \eqref{R1}, \eqref{R2} and \eqref{R3}, and similarly \eqref{L1}, \eqref{L2} and \eqref{L3}, are symmetric in $i$ and $j$, which means that 
\begin{align}
E_{i,i+1} = E_{i+1,i} = s_i \circ E_{i,i+1} = E_{i,i+1} \circ s_i,\\
F_{i,i+1} = F_{i+1,i} = s_i \circ F_{i,i+1} = F_{i,i+1} \circ s_i.
\end{align}
Furthermore, we may also infer that the conditions that specify matrices inducing $E_{i,i+2}$ and $F_{i,i+2}$ are the same as those for $E_{i,i+1}$ and $F_{i,i+1}$ when we swap $i+1$ with $i+2$, hence
\begin{align}
E_{i,i+2} = E_{i,i+1} \circ s_i \circ E_{i,i+1},\\
E_{i,i+2} = F_{i,i+1}\circ s_i \circ F_{i,i+1}.
\end{align}
We now have two canonical types of transformations, $E_{i,j}$ and $F_{i,j}$ generated by one $E_{i,j}$ and $F_{i,j}$ and the symmetric group, $S_{2m+6}$. It is a straightforward calculation to see from \eqref{eijaction} and \eqref{fijaction} that $E_{i,j}$ and $F_{i,j}$ are involutions. Furthermore, we can see that the composition
\[
T_{i,j} = F_{i,j} \circ E_{i,j} \cdot u_i, u_j \to q u_i, qu_j,
\]
is of infinite order, hence, the transformations, $E_{i,j}$ and $F_{i,j}$, generate an infinite dihedral group.

One additional transformation we wish to consider, which we denote $\iota$, is a transformation has a nontrivial action on $\eta$, which by definition is the value such that \eqref{Asym} holds. We claim that $\iota$ is induced by transformation of the form \eqref{Rlinear} where $R(z) = B(\eta/z)$.

\begin{thm}
The transformation induced by \eqref{Rlinear} where $R(x) = B(\eta/z)$ specified by 
\begin{align}
 &\iota \cdot \eta \to q \eta,\\
 &\iota \cdot u_i \to \dfrac{\eta}{u_i}, \\
 &\iota \cdot (x_k : y_k) = (B_{11}(u_k), B_{21}(u_k)).\\
 & \iota \cdot v= w, \hspace{1cm} \iota \cdot w = qv.
\end{align}
\end{thm}

\begin{proof}
If we allow $R(z) = B(\eta/z)$, then by \eqref{asymcomp} then 
\begin{equation}
(\iota A(z)) = B(z) B\left( \dfrac{\eta}{z} \right)^{-1},
\end{equation}
hence, we have 
\begin{equation}
(\iota A(z)) \left(\iota A\left( \dfrac{\eta}{qz} \right) \right) = B(z)  B\left( \dfrac{\eta}{z} \right)^{-1} B\left( \dfrac{\tilde{\eta}}{qz}\right)
\end{equation}
where we have denoted $\tilde{\eta} = \iota \eta = q \eta$. Up to some scalar, we have that
\[
\iota B(z) = N \begin{pmatrix} B_{22}\left( \dfrac{\eta}{z} \right) & -B_{12}\left( \dfrac{\eta}{z} \right)\\
-B_{21}\left( \dfrac{\eta}{z} \right) & B_{11}\left( \dfrac{\eta}{z} \right)
\end{pmatrix} 
\]
where $N$ is some left multiplicative factor chosen so that the images of $B(z)$ at $z = u_{2m+3}$, $u_{2m+4}$ and $u_{2m+5}$ are $\langle (1,1)\rangle$, $\langle (0,1)\rangle$ and $\langle (1,0)\rangle$ respectively. This matrix is clearly singular at $z = \eta/u_i$ for $i = 0, \ldots, 2m+5$. Furthermore, an element of the  kernel at $z = \eta/u_k$ is given by annihilating either the first or second row, i.e.,  
\[
\iota (x_k:y_k) = (B_{11}(u_k): B_{21}(u_k)) = (B_{12}(u_k): B_{22}(u_k)).
\] 
Using the second row is equivalent because $\det B(u_k) = 0$. 
\end{proof}

We may now define the system we wish to call the ellitpic Garnier system. 

\begin{defn}
The elliptic Garnier system is the system of birational transformations generated by $\iota$ and $E_{i,j}$ and $F_{i,j}$ for $(i,j) \in \mathbb{Z}_{2m+6}^2$.
\end{defn}

The group generated by $S_8$, the $E_{i,j}$ and $F_{i,j}$ for $i,j \in \mathbb{Z}_{2m+6}$, is an affine Weyl group of type $D_{2m+6}^{(1)}$. The transformation, $\iota$, in this setting acts as a Dynkin diagram automorphism. The Dynkin diagram for $W(D_{n}^{(1)})$ is shown in Figure \ref{fig}. If we include $\iota$ in the group, we obtain an extended affine Weyl group of type $D_{2m+6}$. This means that the symmetry group is
\[
\Lambda_{2m+7}\rtimes W(D_{2m+6}),
\]
where $\Lambda_{2m+7}$ is an appropriate lattice of rank $2m+7$ having $\mathbb{Z} \times D_{2m+6}$ as an index 2 sublattice. More details on a geometric level concerning this group can be found in \cite{rains:noncomgeom}.

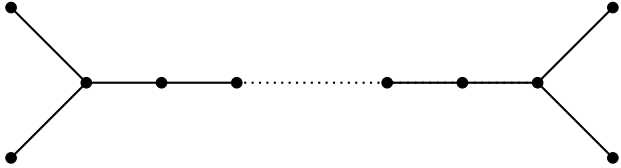
\begin{figure}
\begin{tikzpicture}
\filldraw (0,1) circle(2pt);
\filldraw (1,0) circle(2pt);
\filldraw (0,-1) circle(2pt);
\filldraw (8,-1) circle(2pt);
\filldraw (7,0) circle(2pt);
\filldraw (8,1) circle(2pt);
\draw[thick] (0,1) -- (1,0);
\draw[thick] (0,-1) -- (1,0);
\draw[thick] (8,1) -- (7,0);
\draw[thick] (8,-1) -- (7,0);
\draw[thick] (1,0) -- (3,0);
\filldraw (2,0) circle(2pt);
\filldraw (3,0) circle(2pt);
\draw[thick] (5,0) -- (7,0);
\filldraw (6,0) circle(2pt);
\filldraw (5,0) circle(2pt);
\draw[dotted,thick] (3,0) -- (7,0);
\end{tikzpicture}
\caption{The Dynkin diagram for $D_n^{(1)}$ for $n \geq 4$.\label{fig}}
\end{figure}

\section{The case $m=1$ and the elliptic Painlev\'e equation}\label{ellipticPE8}

The elliptic Painlev\'e equation lies at the top of the Painlev\'e hierarchy \cite{Sakai:Rational}. It is characterized by having a surface of initial conditions obtained by blowing up $(\mathbb{P}^1)^2$ at eight points and an irreducible anticanonical divisor \cite{Sakai:Rational}. Furthermore, the elliptic Painlev\'e equation possesses an elegant geometric description in terms of moving pencils of elliptic curves \cite{Elliptichypergeomtric}. 

It is known that the elliptic Painlev\'e equation possesses a Lax pair \cite{rains:isomonodromy, Yamada2009} whose most explicit construction was inspired by work on Pad\'e approximations and the geometric description of the elliptic Painlev\'e equation \cite{NoumiYamada:ellE8Lax}. Our work has been guided by the recent work on generalized Hitchin systems \cite{Rains2013}, biorthogonal polynomials \cite{rains:isomonodromy} and a recent geometric description for the elliptic Painlev\'e equations in terms of noncommutative geometry \cite{Okounkov2014, rains:noncomgeom}. These works specify that the moduli space of elliptic difference equations satisfying \eqref{C1}, \eqref{C2} and \eqref{C3} for $m=1$ may be identified with the rational surface of initial conditions for the elliptic Painlev\'e equation. One of the novel features of this work is that we are able to be very explicit and demonstrate a range of symmetries directly from the linear problem.

We have mentioned that $B(z)$ is only defined up to constant matrix multiplication on the left, and a constant gauge transformation of $A(z)$ corresponds to multiplication on the right by some constant matrix, hence, we have that up to constant gauge transformations of $A(z)$, $B(z)$ is only defined up to multiplication on the left and right by constant matrices. This means we may additionally fix three kernel vectors, leaving two kernel vectors as parameters in the case of $m=1$; if the kernels of $B(z)$ at $z= u_i$, $u_j$ and $u_k$ are $\langle (x_i,y_i) \rangle$, $\langle (x_j,y_j) \rangle$ and $\langle (x_k,y_k) \rangle$, then so long as these are trivially intersecting subspaces, the kernels of
\begin{equation}\label{rigidify}
\bar{B}(z) = B(z) N_{ijk},
\end{equation} 
 at $z=a_i$, $a_j$ and $a_k$, where $N_{ijk}$ is the matrix 
\begin{equation}
N_{ijk} = \begin{pmatrix}
 x_i \left(x_j y_k-x_k y_j\right) & x_j \left(x_k y_i-x_i y_k\right) \\
 y_i \left(x_j y_k-x_k y_j\right) & y_j \left(x_k y_i-x_i y_k\right),
\end{pmatrix},
\end{equation}
are $\langle (1,0) \rangle$, $\langle (0,1) \rangle$ and $\langle (1,1) \rangle$ respectively. This means that without loss of generality, we may fix three kernels, $B(u_0)$, $B(u_1)$ and $B(u_2)$ so that they are $\langle (1,0)\rangle$, $\langle (0,1) \rangle$ and $\langle (1,1) \rangle$ respectively. The remaining kernels, $\langle (x_3,y_4) \rangle$ and $\langle (x_4,y_4) \rangle$ are sufficient to parameterize $B(z)$. In this specialization, the precise parameterization of $B(z)$ is given by
\begin{subequations}\label{BE8}
\begin{align}
B_{11}(z) =& 
\frac{x_4 y_3 \theta_p \left(\frac{u_1 u_2 u_3 u_6}{L},\frac{z}{u_0},\frac{z}{u_4},\frac{z}{u_6},\frac{u_0 u_4 u_6 z}{L}\right)}{u_1^3 u_2^3 u_3^3 u_6 \theta_p \left(\frac{u_0}{u_1},\frac{u_0}{u_2},\frac{u_0}{u_3},\frac{u_4}{u_1},\frac{u_4}{u_2},\frac{u_4}{u_3},\frac{u_5}{u_6}\right)}
+\frac{x_3 y_4 \theta_p \left(\frac{u_1 u_2 u_4 u_6}{L},\frac{z}{u_0},\frac{z}{u_3},\frac{z}{u_6},\frac{u_0 u_3 u_6 z}{L}\right)}{u_1^3 u_2^3 u_4^3 u_6 \theta_p \left(\frac{u_0}{u_1},\frac{u_0}{u_2},\frac{u_3}{u_1},\frac{u_3}{u_2},\frac{u_0}{u_4},\frac{u_3}{u_4},\frac{u_5}{u_6}\right)} \\
&+\frac{y_3 y_4 \theta_p \left(\frac{u_1 u_3 u_4 u_6}{L},\frac{z}{u_0},\frac{z}{u_2},\frac{z}{u_6},\frac{u_0 u_2 u_6 z}{L}\right)}{u_1^3 u_3^3 u_4^3 u_6 \theta_p \left(\frac{u_0}{u_1},\frac{u_2}{u_1},\frac{u_0}{u_3},\frac{u_2}{u_3},\frac{u_0}{u_4},\frac{u_2}{u_4},\frac{u_5}{u_6}\right)},\nonumber
\end{align}
\vspace{-.7cm}
\begin{align}
B_{12}(z) =&
  -\frac{x_3 y_4 \theta_p \left(\frac{u_0 u_2 u_3 u_6}{L},\frac{z}{u_1},\frac{z}{u_4},\frac{z}{u_6},\frac{u_1 u_4 u_6 z}{L}\right)}{u_0 u_1^3 u_2 u_3 u_4^3 u_6 \theta_p \left(\frac{u_0}{u_1},\frac{u_2}{u_1},\frac{u_3}{u_1},\frac{u_0}{u_4},\frac{u_2}{u_4},\frac{u_3}{u_4},\frac{u_5}{u_6}\right)}
  -\frac{x_4 y_3 \theta_p \left(\frac{u_0 u_2 u_4 u_6}{L},\frac{z}{u_1},\frac{z}{u_3},\frac{z}{u_6},\frac{u_1 u_3 u_6 z}{L}\right)}{u_0 u_1^3 u_2 u_3^3 u_4 u_6 \theta_p \left(\frac{u_0}{u_1},\frac{u_2}{u_1},\frac{u_0}{u_3},\frac{u_2}{u_3},\frac{u_4}{u_1},\frac{u_4}{u_3},\frac{u_5}{u_6}\right)}\\
 & -\frac{x_3 x_4 \theta_p \left(\frac{u_0 u_3 u_4 u_6}{L},\frac{z}{u_1},\frac{z}{u_2},\frac{z}{u_6},\frac{u_1 u_2 u_6 z}{L}\right)}{u_0 u_1^3 u_2^3 u_3 u_4 u_6 \theta_p \left(\frac{u_0}{u_1},\frac{u_0}{u_2},\frac{u_3}{u_1},\frac{u_3}{u_2},\frac{u_4}{u_1},\frac{u_4}{u_2},\frac{u_5}{u_6}\right)},\nonumber
\end{align}
\vspace{-.7cm}
\begin{align}
B_{21}(z) =&
\frac{x_4 y_3 \theta_p \left(\frac{u_1 u_2 u_3 u_7}{L},\frac{z}{u_0},\frac{z}{u_4},\frac{z}{u_7},\frac{u_0 u_4 u_7 z}{L}\right)}{u_1^3 u_2^3 u_3^3 u_7 \theta_p \left(\frac{u_0}{u_1},\frac{u_0}{u_2},\frac{u_0}{u_3},\frac{u_4}{u_1},\frac{u_4}{u_2},\frac{u_4}{u_3},\frac{u_5}{u_7}\right)}
+\frac{x_3 y_4 \theta_p \left(\frac{u_1 u_2 u_4 u_7}{L},\frac{z}{u_0},\frac{z}{u_3},\frac{z}{u_7},\frac{u_0 u_3 u_7 z}{L}\right)}{u_1^3 u_2^3 u_4^3 u_7 \theta_p \left(\frac{u_0}{u_1},\frac{u_0}{u_2},\frac{u_3}{u_1},\frac{u_3}{u_2},\frac{u_0}{u_4},\frac{u_3}{u_4},\frac{u_5}{u_7}\right)}\\
&+\frac{y_3 y_4 \theta_p \left(\frac{u_1 u_3 u_4 u_7}{L},\frac{z}{u_0},\frac{z}{u_2},\frac{z}{u_7},\frac{u_0 u_2 u_7 z}{L}\right)}{u_1^3 u_3^3 u_4^3 u_7 \theta_p \left(\frac{u_0}{u_1},\frac{u_2}{u_1},\frac{u_0}{u_3},\frac{u_2}{u_3},\frac{u_0}{u_4},\frac{u_2}{u_4},\frac{u_5}{u_7}\right)},\nonumber
\end{align}
\vspace{-.7cm}
\begin{align}
B_{22}(z) =&
  -\frac{x_3 y_4 \theta_p \left(\frac{u_0 u_2 u_3 u_7}{L},\frac{z}{u_1},\frac{z}{u_4},\frac{z}{u_7},\frac{u_1 u_4 u_7 z}{L}\right)}{u_0 u_1^3 u_2 u_3 u_4^3 u_7 \theta_p \left(\frac{u_0}{u_1},\frac{u_2}{u_1},\frac{u_3}{u_1},\frac{u_0}{u_4},\frac{u_2}{u_4},\frac{u_3}{u_4},\frac{u_5}{u_7}\right)}
 & -\frac{x_4 y_3 \theta_p \left(\frac{u_0 u_2 u_4 u_7}{L},\frac{z}{u_1},\frac{z}{u_3},\frac{z}{u_7},\frac{u_1 u_3 u_7 z}{L}\right)}{u_0 u_1^3 u_2 u_3^3 u_4 u_7 \theta_p \left(\frac{u_0}{u_1},\frac{u_2}{u_1},\frac{u_0}{u_3},\frac{u_2}{u_3},\frac{u_4}{u_1},\frac{u_4}{u_3},\frac{u_5}{u_7}\right)}\\
 & -\frac{x_3 x_4 \theta_p \left(\frac{u_0 u_3 u_4 u_7}{L},\frac{z}{u_1},\frac{z}{u_2},\frac{z}{u_7},\frac{u_1 u_2 u_7 z}{L}\right)}{u_0 u_1^3 u_2^3 u_3 u_4 u_7 \theta_p \left(\frac{u_0}{u_1},\frac{u_0}{u_2},\frac{u_3}{u_1},\frac{u_3}{u_2},\frac{u_4}{u_1},\frac{u_4}{u_2},\frac{u_5}{u_7}\right)}.\nonumber
\end{align}
\end{subequations}
This also suggests that the canonical translation should act on $u_3$ and $u_4$. One issue is that the actions specified in the previous section do not fix the kernels $\langle (1,0) \rangle$, $\langle (0,1) \rangle$ and $\langle (1,1) \rangle$. To specify a translation, we must first act by $T_{3,4}$ and then act by a further constant matrix that sends $\langle (1,0) \rangle$, $\langle (0,1)\rangle$ and $\langle (1,1)\rangle$ to the new kernel vectors at $B(u_0)$, $B(u_1)$ and $B(u_2)$ respectively. Furthermore, we may consider these as maps on $(\mathbb{P}^1)^2$, since the other kernels are fixed by this step. Let us denote the actions of $E_{3,4}$ and $F_{3,4}$ with an additional normalization step by $\bar{E}_{3,4}$ and $\bar{F}_{3,4}$ respectively. We specify $\bar{T}_{3,4}$ by specifying the actions of $\bar{E}_{3,4}$ and $\bar{F}_{3,4}$, to give
\begin{align*}
\bar{E}_{3,4} \cdot & (\mathbb{P}^1)^2  \to (\mathbb{P}^1)^2,\\
\bar{F}_{3,4} \cdot &(\mathbb{P}^1)^2  \to (\mathbb{P}^1)^2,\\
\bar{T}_{3,4} \cdot &(\mathbb{P}^1)^2  \to (\mathbb{P}^1)^2,\\
:& ((x_3:y_3), (x_4:y_4)) \to  ((\hat{\tilde{x}}_3:\hat{\tilde{y}}_3), (\hat{\tilde{x}}_4:\hat{\tilde{y}}_4)).
\end{align*}
which we claim to be the elliptic Painlev\'e equation. 

\begin{lem}
The action of $\bar{E}_{3,4}$ is given by
\begin{align}\label{E01E8}
&\bar{E}_{3,4} \cdot ((x_3:y_3),(x_4:y_4)) \to ((\tilde{x}_3:\tilde{y}_3),(\tilde{x}_4:\tilde{y}_4)),
\end{align}
where these values are related by
\begin{align}\label{E01E80}
(\tilde{x}_3:\tilde{y}_3)=
& \left( x_4 \theta_p \left(\frac{\eta }{u_1 u_3},\frac{u_3}{u_1}\right) \left(y_3 \left(x_4-y_4\right) \theta_p \left( \frac{\eta }{u_2 u_3},\frac{u_3}{u_2},\frac{\eta }{u_0 u_4},\frac{u_4}{u_0}\right)\right. \right.\\
&\left.\left.+\left(y_3-x_3\right) y_4 \theta_p \left(\frac{\eta }{u_0 u_3},\frac{u_3}{u_0},\frac{\eta }{u_2 u_4},\frac{u_4}{u_2}\right)\right) \right.\nonumber\\
   &: y_4 \theta_p \left(\frac{\eta }{u_0 u_3},\frac{u_3}{u_0}\right) \left(x_3 \left(x_4-y_4\right) \theta_p \left(\frac{\eta }{u_2 u_3},\frac{u_3}{u_2},\frac{\eta }{u_1 u_4},\frac{u_4}{u_1}\right) \right.\nonumber\\
   & \left.\left.+x_4 \left(y_3-x_3\right) \theta_p \left(\frac{\eta }{u_1 u_3},\frac{u_3}{u_1},\frac{\eta }{u_2 u_4},\frac{u_4}{u_2}\right)\right) \right)\nonumber
\end{align}
and
\begin{align}\label{E01E81}
(\tilde{x}_4:\tilde{y}_4)=
&\left(  x_3 \theta_p \left(\frac{\eta }{u_1 u_4},\frac{u_4}{u_1}\right) \left(y_3 \left(y_4-x_4\right) \theta_p \left(\frac{\eta }{u_2 u_3},\frac{u_3}{u_2},\frac{\eta }{u_0 u_4},\frac{u_4}{u_0}\right)\right.\right.\\
& \left.+\left(x_3-y_3\right) y_4 \theta_p \left(\frac{\eta }{u_0 u_3},\frac{u_3}{u_0},\frac{\eta }{u_2 u_4},\frac{u_4}{u_2}\right)\right) \nonumber\\
&\left.: y_3 \theta_p \left(\frac{\eta }{u_0 u_4},\frac{u_4}{u_0}\right) \left(x_3 \left(y_4-x_4\right) \theta_p \left(\frac{\eta }{u_2 u_3},\frac{u_3}{u_2},\frac{\eta }{u_1 u_4},\frac{u_4}{u_1}\right)\right.\right.\nonumber\\
&\left.\left.+x_4 \left(x_3-y_3,\frac{\eta }{u_1 u_3},\frac{u_3}{u_1},\frac{\eta }{u_2 u_4},\frac{u_4}{u_2}\right)\right) \right)\nonumber
\end{align}
\end{lem}

\begin{proof}
We first compute the action of $E_{3,4}$ on all the kernels, which we may write in terms of a map on $(\mathbb{P}^1)^5$, where each coordinate represents the kernels of $B(u_0)$, \ldots, $B(u_5)$. This action of $E_{3,4}$ is as follows
\begin{align*}
&((1:0),(0:1),(1:1),(x_3:y_3),(x_4:y_4))  \to \\
& \left(\left(u_4 x_3 y_4 \theta_p \left(\frac{\eta }{u_0 u_3},\frac{u_3}{u_0},\frac{v}{u_4},\frac{\eta }{v u_4}\right) -u_3 x_4 y_3 \theta_p \left(\frac{v}{u_3},\frac{\eta }{v u_3},\frac{\eta }{u_0 u_4},\frac{u_4}{u_0}\right)\right. \right.\\
& \left.:y_3 y_4 \left(u_4 \theta_p \left(\frac{\eta }{u_0 u_3},\frac{u_3}{u_0},\frac{v}{u_4},\frac{\eta }{v u_4}\right)-u_3 \theta_p \left(\frac{v}{u_3},\frac{\eta }{v u_3},\frac{\eta }{u_0 u_4},\frac{u_4}{u_0}\right)\right)\right),\\
&  \left( x_3 x_4 \left(u_3 \theta_p \left(\frac{v}{u_3},\frac{\eta }{v u_3},\frac{\eta }{u_1 u_4},\frac{u_4}{u_1}\right)-u_4 \theta_p \left(\frac{\eta }{u_1 u_3},\frac{u_3}{u_1},\frac{v}{u_4},\frac{\eta }{v u_4}\right)\right) \right.\\
& \left. : u_3 x_3 y_4 \theta_p \left(\frac{v}{u_3},\frac{\eta }{v u_3},\frac{\eta }{u_1 u_4},\frac{u_4}{u_1}\right)-u_4 x_4 y_3 \theta_p \left(\frac{\eta }{u_1 u_3},\frac{u_3}{u_1},\frac{v}{u_4},\frac{\eta}{v u_4}\right) \right),\\
&\left(  u_4 x_3 \left(y_4-x_4\right) \theta_p \left(\frac{\eta }{u_2 u_3},\frac{u_3}{u_2},\frac{v}{u_4},\frac{\eta }{v u_4}\right)+u_3 x_4 \left(x_3-y_3\right) \theta_p \left(\frac{v}{u_3},\frac{\eta }{v u_3},\frac{\eta }{u_2 u_4},\frac{u_4}{u_2}\right) \right.\\
&\left.:u_4 y_3 \left(y_4-x_4\right)\theta_p \left(\frac{\eta }{u_2 u_3},\frac{u_3}{u_2},\frac{v}{u_4},\frac{\eta }{v u_4}\right)+u_3 \left(x_3-y_3\right) y_4 \theta_p \left(\frac{v}{u_3},\frac{\eta }{v u_3},\frac{\eta }{u_2 u_4},\frac{u_4}{u_2}\right)\right),\\
&(x_4:y_4), (x_3:y_3))
\end{align*}
Using these kernels new for $B(u_0)$, $B(u_1)$ and $B(u_2)$ in $N_{012}$ above is the relevent constant matrix. Since $E_{3,4}$ itself switches $(x_3:y_3)$ and $(x_4:y_4)$, we obtain an element of the new kernel of $\tilde{B}(u_3)$ by solving
\[
\begin{pmatrix} x_4 \\ y_4 \end{pmatrix} = N_{012} \begin{pmatrix} \tilde{x}_3 \\ \tilde{y}_3 \end{pmatrix},
\] 
which is equivalent to \eqref{E01E80} and similarly, the same step for $z = u_4$ gives \eqref{E01E81}. 
\end{proof}

In calculating the remaining transformation, $\bar{F}_{3,4}$, we know that $F_{3,4}$ does not change kernels of $B(u_0)$, $B(u_1)$ and $B(u_2)$, hence, we have the following result. 

\begin{lem}
The action of $\bar{F}_{3,4}$ is given by
\begin{align}
\bar{F}_{3,4} \cdot
(\hat{x}_3:\hat{y}_3) \to &\left(B_{21}(u_3)B_{21}\left(\dfrac{\eta}{qu_3}\right)-B_{11}(u_3)B_{22}\left(\dfrac{\eta}{qu_3} \right):\right. \\ 
&\left. B_{11}(u_3) B_{21}\left( \dfrac{\eta}{qu_3} \right)- B_{21}(u_3)B_{11}\left(\dfrac{\eta}{qu_3} \right)\right)  \nonumber \\
(\hat{x}_4:\hat{y}_4) = &\left(B_{21}(u_4)B_{21}\left(\dfrac{\eta}{qu_4}\right)-B_{11}(u_4)B_{22}\left(\dfrac{\eta}{qu_4} \right):\right. \\ 
&\left. B_{11}(u_4) B_{21}\left( \dfrac{\eta}{qu_4} \right)- B_{21}(u_)B_{11}\left(\dfrac{\eta}{qu_j} \right)\right),  \nonumber
\end{align}
where the entries are given by \eqref{BE8}.
\end{lem}

\begin{proof}
This follows from the fact that $F_{3,4}$ itself does not change the kernels, as it is induced by multiplication on the left of $B(z)$. The only kernels changed are those of $B(u_3)$ and $B(u_4)$, and the calculation follows from the previous section.
\end{proof}

We may now construct the explicit Lax pair for the elliptic Painlev\'e equation.

\begin{thm}\label{EllP}
The elliptic Painlev\'e equation, arises as the compatibility between \eqref{Alinear} where $A(z)$ is specified by \eqref{product}, \eqref{Bentries} and \eqref{BE8}, and \eqref{Rlinear}, where $R(z)$ is given by
\[
R(z) = ((F_{3,4} R_l(z))(F_{3,4} N_{012}))^{-1}, 
\]
where $R_l(z)$ is given by \eqref{RL} for $i= 3$ and $j = 4$.
\end{thm}

\begin{figure}
\begin{tikzpicture}
\filldraw (0,0) circle(2pt);
\filldraw (1,0) circle(2pt);
\filldraw (2,0) circle(2pt);
\filldraw (2,1) circle(2pt);
\filldraw (3,0) circle(2pt);
\filldraw (4,0) circle(2pt);
\filldraw (5,0) circle(2pt);
\filldraw (6,0) circle(2pt);
\filldraw (7,0) circle(2pt);
\draw[thick] (0,0) -- (7,0);
\draw[thick] (2,0)--(2,1);
\end{tikzpicture}
\caption{The Dynkin diagram for $E_8^{(1)}$.\label{E8fig}}
\end{figure}
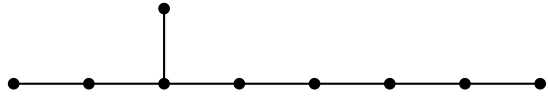

The group $W(D_8^{(1)})$ with an additional Dynkin diagram automorphism coincides with $(\mathbb{Z}\times \Lambda_{E_8})\rtimes W(D_8)$. The Dynkin diagram of type $E_8^{(1)}$ is shown in Figure \ref{E8fig}. We now present the system of isomonodromic deformations in the $m=1$ case in a parameterization that will make the symmetries more apparent. First we can specialize the coordinates so that $x_3 = 1$ and $y_3 = f$ so that $f$ is treated as an affine coordinate. Secondly, the image of $B(u_3)$, which we parameterize as $\langle (1,g) \rangle$ with $g$ treated as an affine coordinate, is related to the kernel of $B(u_4)$ by
\begin{align}\label{kerBu4}
\dfrac{y_4}{x_4} = \frac{\theta_p \left(\frac{u_1}{u_2},\frac{u_0}{u_4}\right) \left(g u_7 \theta_p \left(\frac{u_3}{u_6},\frac{u_5}{u_7},\frac{u_1 u_2 u_3 u_6}{L},\frac{u_0 u_3 u_4 u_6}{L}\right)-u_6 \theta_p \left(\frac{u_5}{u_6},\frac{u_3}{u_7},\frac{u_1 u_2 u_3 u_7}{L},\frac{u_0 u_3 u_4
   u_7}{L}\right)\right)}{\theta_p \left(\frac{u_0}{u_2},\frac{u_1}{u_4}\right) \left(g u_7 \theta_p \left(\frac{u_3}{u_6},\frac{u_5}{u_7},\frac{u_0 u_2 u_3 u_6}{L},\frac{u_1 u_3 u_4 u_6}{L}\right)-u_6 \theta_p \left(\frac{u_5}{u_6},\frac{u_3}{u_7},\frac{u_0 u_2 u_3 u_7}{L},\frac{u_1 u_3 u_4
   u_7}{L}\right)\right)}.
\end{align}
This means that we may either parameterize $B(z)$ in terms of five kernels and three images or we can parameterize $B(z)$ using four kernels and four images. It follows that $A(z)$ is uniquely specified by the following data:
\begin{align}
\tag{$P_1$} \label{P1}&\begin{array}{p{15cm}} The determinant of $B(z)$ vanishes at points $z=u_i$ for $i=0,\ldots,7$ and is nonzero.\end{array}\\
\tag{$P_2$} \label{P2}&\begin{array}{p{15cm}} The kernels of $B(z)$ at $z = u_0, u_1, u_2$ and $u_3$ are $\langle (1,0) \rangle$, $\langle (0,1) \rangle$, $\langle (1,1) \rangle$ and $\langle (1,f) \rangle$ respectively.\end{array}\\
\tag{$P_3$} \label{P3}&\begin{array}{p{15cm}} The images of $B(z)$ at $z = u_7, u_6, u_5$ and $u_3$ are $\langle (1,0) \rangle$, $\langle (0,1) \rangle$, $\langle (1,1) \rangle$ and $\langle (1,g) \rangle$ respectively.\end{array}
\end{align}
The action of the symmetric group, $S_8$, on the $u_i$ variables only requires renormalization, which is induced by constant matrix multiplication either on the left or right. We use the notation
\[
X = \left( \begin{array}{c c c c}
u_0 & u_1 & u_2 & u_3 \\
u_4 & u_5 & u_6 & u_7 \end{array} ; \eta, f,g\right),
\]
to denote the space of parameters and the affine coordinates specifying the image and the kernels of $B(u_3)$. 

The first generator we specify is $s_0$ which is induced by right multiplication by a constant matrix with $\langle( 1,1)\rangle$ as an eigenspace and that permutes  $\langle (1,0) \rangle$ and $\langle (0,1) \rangle$. One such matrix is the matrix with ones on the off diagonal entries, in which the preimage of $(1,f)$ is $(f,1)$, hence, this induces a transformation that sends $f$ to $1/f$. The same action is valid on the right, in which case this matrix permutes the images of $B(u_6)$ and $B(u_7)$ defining $s_6$. These symmetries are given by
\begin{align}
s_0 &\cdot X = \left( \begin{array}{c c c c}
u_1 & u_0 & u_2 & u_3 \\
u_4 & u_5 & u_6 & u_7 \end{array} ; \eta, \dfrac{1}{f} ,g\right), \\
s_6 &\cdot X=  \left( \begin{array}{c c c c}
u_0 & u_1 & u_2 & u_3 \\
u_4 & u_5 & u_7 & u_6 \end{array} ; \eta, f ,\dfrac{1}{g} \right). 
\end{align}
Similarly, the matrix with $\langle (1,0) \rangle$ as an eigenspace that permutes $\langle (1,1) \rangle$ and $\langle (0,1) \rangle$ has the effect
\begin{align}
s_1 &\cdot X = \left( \begin{array}{c c c c}
u_0 & u_2 & u_1 & u_3 \\
u_4 & u_5 & u_6 & u_7 \end{array} ; \eta, \dfrac{f}{1-f} ,g\right), \\
s_5 &\cdot X= \left( \begin{array}{c c c c}
u_0 & u_1 & u_2 & u_3 \\
u_4 & u_6 & u_5 & u_7 \end{array} ; \eta, f ,\dfrac{g}{1-g} \right).
\end{align}
The actions $s_2$ and $s_4$ are induced by diagonal matrices since the matrix possesses eigenspaces $\langle (1,0)\rangle$ and $\langle (0,1)\rangle$. The action of $s_2$ is induced by a matrix that sends $(1,f)$ to $\langle (1,1)\rangle$ on the right, i.e., $\mathrm{diag}(1,1/f)$, however, in swapping $u_2$ and $u_3$, we also change $g$ to correspond to the image of $B(u_2)$, hence, we have
\begin{equation}
s_2 \cdot X =  \left( \begin{array}{c c c c}
u_0 & u_1 & u_3 & u_2 \\
u_4 & u_5 & u_6 & u_7 \end{array} ; \eta, \dfrac{1}{f} , \dfrac{B_{21}(u_2)}{B_{11}(u_2)} =: \bar{g} \right).
\end{equation}
where $\bar{g}$ is related to $f$ and $g$ by 
\begin{align}
\dfrac{\theta_p\left( \frac{u_1}{u_3},\frac{u_2}{u_4} \right)}{\theta_p\left( \frac{u_2}{u_3},\frac{u_1}{u_4} \right)} f = & \frac{\bar{g} u_7 \theta_p \left(\frac{u_2}{u_6},\frac{u_5}{u_7},\frac{u_0 u_2 u_3 u_6}{L},\frac{u_1 u_2 u_4 u_6}{L}\right)-u_6 \theta_p \left(\frac{u_5}{u_6},\frac{u_2}{u_7},\frac{u_0 u_2 u_3 u_7}{L},\frac{u_1 u_2 u_4 u_7}{L}\right)}{\bar{g} u_7 \theta_p \left(\frac{u_2}{u_6},\frac{u_5}{u_7},\frac{u_1 u_2 u_3 u_6}{L},\frac{u_0 u_2 u_4 u_6}{L}\right)-u_6 \theta_p \left(\frac{u_5}{u_6},\frac{u_2}{u_7},\frac{u_1 u_2 u_3 u_7}{L},\frac{u_0 u_2 u_4 u_7}{L}\right)}\\
   &\times \frac{g u_7 \theta_p \left(\frac{u_3}{u_6},\frac{u_5}{u_7},\frac{u_1 u_2 u_3 u_6}{L},\frac{u_0 u_3 u_4 u_6}{L}\right)-u_6 \theta_p \left(\frac{u_5}{u_6},\frac{u_3}{u_7},\frac{u_1 u_2 u_3 u_7}{L},\frac{u_0 u_3 u_4 u_7}{L}\right)}{g u_7 \theta_p \left(\frac{u_3}{u_6},\frac{u_5}{u_7},\frac{u_0 u_2 u_3 u_6}{L},\frac{u_1 u_3 u_4 u_6}{L}\right)-u_6 \theta_p \left(\frac{u_5}{u_6},\frac{u_3}{u_7},\frac{u_0 u_2 u_3 u_7}{L},\frac{u_1 u_3 u_4 u_7}{L}\right)}.\nonumber
\end{align}
To continue to specify a standard set of generators for $S_8$, the action of $s_3$ swaps $u_3$ and $u_4$ which must swap the kernels and images of $B(u_3)$ with those of $B(u_4)$, specifically
\begin{align}
s_3 \cdot& X = \left( \begin{array}{c c c c}
u_0 & u_1 & u_2 & u_4 \\
u_3 & u_5 & u_6 & u_7 \end{array} ; \eta, \dfrac{y_4}{x_4}, \dfrac{B_{21}(u_4)}{B_{11}(u_4)} \right),\\
&\dfrac{B_{21}(u_4)}{B_{11}(u_4)} = \frac{u_6 \theta_p \left(\frac{u_5}{u_6},\frac{u_4}{u_7}\right) \left(\theta_p \left(\frac{u_1}{u_2},\frac{u_0}{u_3},\frac{u_1 u_2 u_4 u_7}{L},\frac{u_0 u_3 u_4 u_7}{L}\right)-f \theta_p \left(\frac{u_0}{u_2},\frac{u_1}{u_3},\frac{u_0 u_2 u_4 u_7}{L},\frac{u_1 u_3 u_4 u_7}{L}\right)\right)}
{u_7
 \theta_p \left(\frac{u_4}{u_6},\frac{u_5}{u_7}\right)\left( \theta_p \left(\frac{u_1}{u_2},\frac{u_0}{u_3},\frac{u_1 u_2 u_4 u_6}{L},\frac{u_0
   u_3 u_4 u_6}{L}\right)-f \theta_p \left(\frac{u_0}{u_2},\frac{u_1}{u_3},\frac{u_0 u_2 u_4 u_6}{L},\frac{u_1 u_3 u_4 u_6}{L}\right)\right)}, \label{ImBu4},
\end{align}
with $y_4/x_4$ as stated by \eqref{kerBu4}. Lastly, $s_4$ is induced by right multiplication by a diagonal matrix that sends the image of $B(u_4)$ to $\langle (1,1) \rangle$. Using the image of $(0,1)$, it is easy to see that a matrix that does this is given by $\mathrm{diag}(1,B_{21}(u_4)/B_{11}(u_4))$, giving the transformation
\[
s_4 \cdot X =  \left( \begin{array}{c c c c}
u_0 & u_2 & u_1 & u_3 \\
u_4 & u_5 & u_6 & u_7 \end{array} ; \eta, f, g\dfrac{B_{21}(u_4)}{B_{11}(u_4)} \right).
\]
This completes the birational representation of the full symmetric group, $S_8$. 

It is important to know where the base points are which are required to make the action of $S_8$ an automorphism of some surface. This can be done by determining which points, if any, arise as images of lines. For example, we may consider the image of $s_4$ of points of the form
\begin{align*}
s_4 \cdot  \left( \begin{array}{c c c c}
u_0 & u_1 & u_2 & u_3 \\
u_4 & u_5 & u_6 & u_7 \end{array} ; \eta, \frac{\theta_p \left(\frac{u_1}{u_2},\frac{u_0}{u_3},\frac{u_1 u_2 u_4 u_7}{L},\frac{u_0 u_3 u_4 u_7}{L}\right)}{\theta_p \left(\frac{u_0}{u_2},\frac{u_1}{u_3},\frac{u_0 u_2 u_4 u_7}{L},\frac{u_1 u_3 u_4 u_7}{L}\right)}, g \right) \\
=  \left( \begin{array}{c c c c}
u_0 & u_1 & u_2 & u_3 \\
u_5 & u_4 & u_6 & u_7 \end{array} ; \eta, \frac{\theta_p \left(\frac{u_1}{u_2},\frac{u_0}{u_3},\frac{u_1 u_2 u_4 u_7}{L},\frac{u_0 u_3 u_4 u_7}{L}\right)}{\theta_p \left(\frac{u_0}{u_2},\frac{u_1}{u_3},\frac{u_0 u_2 u_4 u_7}{L},\frac{u_1 u_3 u_4 u_7}{L}\right)}, 0 \right),
\end{align*}
where $g$ is any element, hence, the inverse of this point is not defined. To make the action of $s_4$ regular, we need to blow up the point
\[
P_1 =  \left( \frac{\theta_p \left(\frac{u_1}{u_2},\frac{u_0}{u_3},\frac{u_1 u_2 u_4 u_7}{L},\frac{u_0 u_3 u_4 u_7}{L}\right)}{\theta_p \left(\frac{u_0}{u_2},\frac{u_1}{u_3},\frac{u_0 u_2 u_4 u_7}{L},\frac{u_1 u_3 u_4 u_7}{L}\right)},0 \right),
\]
which is presented by affine coordinates, $(f,g)$, representing points in $(\mathbb{P}^1)^2$. This also coincides with the point in which the images of $B(u_3)$, $B(u_4)$ and $B(u_7)$ are equal, which is a natural point to consider since if we were to swap the roles of $u_3$ and $u_4$ with $u_5$ and $u_6$, then we cannot find a matrix which sends the images of $B(u_5)$ and $B(u_6)$ to $(1,1)$ and $(0,1)$ respectively. We simiilarly find that $s_4$ is ill-defined on the point
\[
P_2 = \left( \frac{\theta_p \left(\frac{u_1}{u_2},\frac{u_0}{u_3},\frac{u_1 u_2 u_4 u_6}{L},\frac{u_0 u_3 u_4 u_6}{L}\right)}{\theta_p \left(\frac{u_0}{u_2},\frac{u_1}{u_3},\frac{u_0 u_2 u_4 u_6}{L},\frac{u_1 u_3 u_4 u_6}{L}\right)},\infty\right),
\]
and since $s_4$ fixes $f$ and replaces $g$ by a ratio of bilinear forms, blowing up these two points suffices to make $s_4$ regular. However, on the blowup $s_5$ is no longer regular since it does not permute the two points we blew up. This forces us to blow up a third point, and similar consider considerations for $s_2$ give us an additional three points.  To specify these points we utilize the fact that all the points lie in the image of an elliptic curve in $(\mathbb{P}^1)^2$, in fact, we find every point appears in the form
\[
\chi(z) = \left( \dfrac{\theta_p \left( \frac{z}{u_1u_2}, \frac{u_1}{u_2}, \frac{z}{u_0u_3},\frac{u_0}{u_3} \right)}{\theta_p\left( \frac{z}{u_0u_2},\frac{u_0}{u_2},\frac{z}{u_0u_3},\frac{u_0}{u_3}\right)},
 \dfrac{u_6 \theta_p\left( \frac{L}{zu_5u_6}, \frac{u_5}{u_6}, \frac{zu_3u_7}{L}, \frac{u_3}{u_7}\right)}{u_7 \theta_p\left( \frac{L}{zu_5u_7}, \frac{u_5}{u_7}, \frac{zu_3u_6}{L}, \frac{u_3}{u_6}\right)}\right),
\]
for certain values of $z$. In particular, we have
\begin{align*}
&P_1 = \chi\left( \frac{L}{u_5u_6}\right), \hspace{1cm} &&P_2 = \chi\left( \frac{L}{u_5u_7}\right) \hspace{1cm}&&P_3 =\chi\left( \frac{L}{u_6u_7}\right),\\
&P_4 = \chi\left(u_1u_2\right), \hspace{1cm}&&P_5 = \chi\left( u_0u_2\right), \hspace{1cm}&&P_6 = \chi\left(u_0u_1\right).
\end{align*}
In fact blowing up these six points makes the action of $S_8$ holomorphic. This configuration of points is preserved under the action of $s_0$, $s_1$, $s_3$, $s_5$ and $s_6$, so those remain holomorphic on the blow-up. Blowing up $P_1$ and $P_2$ makes $s_4$ holomorphic, and it permutes the other four points, so it remains holomorphic on the full blow-up. Similarly $s_2$ permutes the points that were not needed to make it holomorphic and hence is holomorphic. 

With the action of the symmetry group defined what remains to be done is that we consider $\iota$ and any nontrivial transformation on both the left and right. The action $\iota$ is simple in this setting, as the transformation can simply identify images with kernels and kernels with images up to some permutation. This means that the discrete isomonodromic deformation induced by $R(z) = B(1/z)$ may be given by
\[
\iota \cdot X = \left( \begin{array}{c c c c}
\frac{\eta}{u_7} & \frac{\eta}{u_6} & \frac{\eta}{u_5} & \frac{\eta}{u_3} \\
\frac{\eta}{u_4} & \frac{\eta}{u_2} & \frac{\eta}{u_1} & \frac{\eta}{u_0} \end{array} ; q\eta, g,f\right).
\]
It is also simple to describe the action of two nontrivial isomonodromic deformations, namely $\bar{E}_{0,1}$ and $\bar{F}_{6,7}$. When we specialize kernel vectors to $\langle (1,0)\rangle$ and $\langle (0,1)\rangle$, then property \eqref{R2} dictates that $R_r(z)$ is diagonal, then since diagonal matrices preserve $\langle (1,0)\rangle$ and $\langle (0,1)\rangle$ but possibly change $\langle (1,1) \rangle$, the matrix fixing $\langle (1,0)\rangle$ and $\langle (0,1)\rangle$ and that sends the preimage of $(1,1)$ under $R(u_2)$ to $\langle (1,1) \rangle$ is also diagonal. This means that if $R_r(z)$ is the matrix inducing $E_{0,1}$ and $\bar{R}_r(z)$ induces $\bar{E}_{0,1}$, then 
\[
\bar{R}_r(z) = R_r(z)R_r(u_2)^{-1} = \begin{pmatrix} \frac{\theta_p\left(\frac{z}{u_1},\frac{\eta}{zu_1}\right)}{\theta_p\left(\frac{u_2}{u_1},\frac{\eta}{u_2u_1}\right)} & 0 \\ 0 & \frac{\theta_p\left(\frac{z}{u_0},\frac{\eta}{zu_0}\right)}{\theta_p\left(\frac{u_2}{u_0},\frac{\eta}{u_2u_0}\right)} \end{pmatrix},
\]
inducing the transformation
\begin{align}
\bar{E}_{0,1} \cdot X \to  \left( \begin{array}{c c c c}
\frac{\eta}{u_0} & \frac{\eta}{u_1} & u_2 & u_3 \\
u_4 & u_5 & u_6 & u_7 \end{array} ; \eta, f \frac{\theta_p \left(\frac{\eta }{u_0 u_2},\frac{u_2}{u_0},\frac{\eta }{u_1 u_3},\frac{u_3}{u_1}\right)}{\theta_p \left(\frac{\eta }{u_1 u_2},\frac{u_2}{u_1},\frac{\eta }{u_0 u_3},\frac{u_3}{u_0}\right)} ,g\right),
\end{align}
Similarly, on the left we have the images of $B(u_6)$ and $B(u_7)$ are  $\langle (0,1)\rangle$ and $\langle (1,0)\rangle$, hence, by \eqref{L2}, we have that if $R_l(z)$ is the matrix inducing $F_{6,7}$ and $\bar{R}_l(z)$ induces $\bar{F}_{6,7}$, then 
\[
\bar{R}_l(z) = R_l(u_5)^{-1} R_l(z) =  \begin{pmatrix} \frac{\theta_p\left(\frac{z}{u_7},\frac{\eta}{zu_7}\right)}{\theta_p\left(\frac{u_5}{u_7},\frac{\eta}{u_5u_7}\right)} & 0 \\ 0 & \frac{\theta_p\left(\frac{z}{u_6},\frac{\eta}{zu_6}\right)}{\theta_p\left(\frac{u_5}{u_6},\frac{\eta}{u_5u_6}\right)} \end{pmatrix},
\]
inducing the transformation
\begin{align}
\bar{F}_{6,7} \cdot X \to  \left( \begin{array}{c c c c}
u_0 & u_1 & u_2 & u_3 \\
u_4 & u_5 & \frac{\eta}{qu_6} & \frac{\eta}{qu_7} \end{array} ; \eta, f ,g\frac{\theta_p \left(\frac{\eta }{u_3 u_6},\frac{u_3}{u_6},\frac{\eta }{u_5 u_7},\frac{u_5}{u_7}\right)}{\theta_p \left(\frac{\eta }{u_5 u_6},\frac{u_5}{u_6},\frac{\eta }{u_3 u_7},\frac{u_3}{u_7}\right)}\right).
\end{align}
This means that the nontrivial part of the isomonodromic deformations are contained within the symmetries of $A(z)$. There is just one more transformation, we obtain the same moduli space of matrices if we replace $\eta$ by $qL/\eta$. This has a trivial effect on $B(z)$, but changes the symmetry of the equation. This is no longer a gauge transformation, but rather comes from the application of a certain elliptic version of the Fourier-Laplace transform to the solutions of the equation \cite{rains:noncomgeom}. Including this additional symmetry extends the group to the full $\mathbb{Z} \times W(E^{(1)}_8)$ symmetry coming from the geometry (it corresponds to the node of the Dynkin diagram for $E^{(1)}_8$ which when removed gives the Dynkin diagram for $D_8$). 

We now note that the configuration of six points we blew up above is not preserved by the generators not in $S_8$ and thus we must blow up two additional points
\[
P_7 = \chi(\eta), \hspace{2cm} P_8 = \chi\left(\frac{qL}{\eta}\right).
\]
These two new points are both preserved by $S_8$, hence, the action of $S_8$ remains holomorphic, while the remaining generators were regular on $(\mathbb{P}^1)^2$ and permute the eight base points. It follows therefore that the full group acts holomorphically on the blow up of $(\mathbb{P}^1)^2$ at eight points. By comparison with \cite{Sakai:Rational} we find that this is indeed the elliptic Painlev\'e equation, completing the proof of Theorem \ref{EllP}.

\section{Discussion}

We have presented a system that may be regarded as an elliptic analogue of the Garnier system. Since the case of $m=1$ corresponds to the elliptic Painlev\'e equation, this system sits above each of the Painlev\'e equations. The lowest member, the case in which $m=0$ is rigid, must correspond to the elliptic hypergeometric equation (cf. \cite{Rains2009}). 
 
\section*{Acknowledgements}

The work of EMR was partially supported by the National Science Foundation under the grant DMS-1500806.

\bibliographystyle{plain}%

\bibliography{C:/Mathematics/TeX/refs}

\def\cprime{$'$}
\begin{thebibliography}{10}

\bibitem{Birkhoff}
G.~D. Birkhoff.
\newblock General theory of linear difference equations.
\newblock {\em Trans. Amer. Math. Soc.}, 12(2):243--284, 1911.

\bibitem{Birkhoffallied}
G.~D. Birkhoff.
\newblock The generalized {R}iemann problem for linear differential equations
  and the allied problems for linear difference and $q$-difference equations.
\newblock {\em Proc. Amer. Acad.}, 49:512--568, 1913.

\bibitem{BirkhoddAdamsSum}
G.~D. Birkhoff and P.~E. Guenther.
\newblock Note on a canonical form for the linear {$q$}-difference system.
\newblock {\em Proc. Nat. Acad. Sci. U. S. A.}, 27:218--222, 1941.

\bibitem{Borodin:connection}
A.~Borodin.
\newblock Isomonodromy transformations of linear systems of difference
  equations.
\newblock {\em Ann. of Math. (2)}, 160(3):1141--1182, 2004.

\bibitem{Etingof1995}
P.~Etingof.
\newblock Galois groups and connection matrices of -difference equations.
\newblock {\em Electronic Research Announcements of the American Mathematical
  Society}, 1(1):1--9, 1995.

\bibitem{Fuchs2}
R.~Fuchs.
\newblock \"{U}ber lineare homogene {D}ifferentialgleichungen zweiter {O}rdnung
  mit drei im {E}ndlichen gelegenen wesentlich singul\"aren {S}tellen.
\newblock {\em Math. Ann.}, 63(3):301--321, 1907.

\bibitem{Fuchs1}
R.~Fuchs.
\newblock \"{U}ber lineare homogene {D}ifferentialgleichungen zweiter {O}rdnung
  mit drei im {E}ndlichen gelegenen wesentlich singul\"aren {S}tellen.
\newblock {\em Math. Ann.}, 70(4):525--549, 1911.

\bibitem{Garnier}
R.~Garnier.
\newblock Sur des \'equations diff\'erentielles du troisi\`eme ordre dont
  l'int\'egrale g\'en\'erale est uniforme et sur une classe d'\'equations
  nouvelles d'ordre sup\'erieur dont l'int\'egrale g\'en\'erale a ses points
  critiques fixes.
\newblock {\em Ann. Sci. \'Ecole Norm. Sup. (3)}, 29:1--126, 1912.

\bibitem{GasperRahman}
G.~Gasper and M.~Rahman.
\newblock {\em Basic hypergeometric series}, volume~35 of {\em Encyclopedia of
  Mathematics and its Applications}.
\newblock Cambridge University Press, Cambridge, 1990.
\newblock With a foreword by Richard Askey.

\bibitem{Sakai:qP6}
M.~Jimbo and H.~Sakai.
\newblock A {$q$}-analog of the sixth {P}ainlev\'e equation.
\newblock {\em Lett. Math. Phys.}, 38(2):145--154, 1996.

\bibitem{Elliptichypergeomtric}
K.~Kajiwara, T.~Masuda, M.~Noumi, Y.~Ohta, and Y.~Yamada.
\newblock {${}\sb {10}E\sb 9$} solution to the elliptic {P}ainlev\'e equation.
\newblock {\em J. Phys. A}, 36(17):L263--L272, 2003.

\bibitem{Krichever2004}
I.~M. Krichever.
\newblock Analytic theory of difference equations with rational and elliptic
  coefficients and the riemann-hilbert problem.
\newblock {\em Russian Mathematical Surveys}, 59(6):1117, 2004.

\bibitem{Mumford1983a}
D.~Mumford.
\newblock Tata lectures on theta i (progress in mathematics vol 28), 1983.

\bibitem{Mumford1983}
D.~Mumford.
\newblock Tata lectures on theta. i, volume 28 of progress in mathematics,
  1983.

\bibitem{Nijhoff2016}
F.~W. Nijhoff and N.~Delice.
\newblock On elliptic lax pairs and isomonodromic deformation systems for
  elliptic lattice equations.
\newblock {\em arXiv preprint arXiv:1605.00829}, 2016.

\bibitem{NoumiYamada:ellE8Lax}
M.~{Noumi}, S.~{Tsujimoto}, and Y.~{Yamada}.
\newblock {Pad\'e interpolation for elliptic {P}ainlev\'e equation}.
\newblock {\em ArXiv e-prints}, April 2012.
\newblock arxiv:1204.0294.

\bibitem{Okamoto1981}
K.~Okamoto.
\newblock {\em Isomonodromic deformation and Painlev{\'e} equations, and the
  Garnier system}.
\newblock Universit{\'e} Louis Pasteur, Institut de Recherche Math{\'e}matique
  Avanc{\'e}e, 1981.

\bibitem{Okounkov2014}
Andrei Okounkov and Eric Rains.
\newblock Noncommutative geometry and painlev$\backslash$'e equations.
\newblock {\em arXiv preprint arXiv:1404.5938}, 2014.

\bibitem{Ormerodlattice}
C.~M. Ormerod.
\newblock The lattice structure of connection preserving deformations for
  $q$-{P}ainlev\'e equations {I}.
\newblock {\em SIGMA Symmetry Integrability Geom. Methods Appl.}, 7:Paper 045,
  22, 2011.

\bibitem{Ormerod2016}
C.~M Ormerod and E.~M. Rains.
\newblock Commutation relations and discrete {G}arnier systems.
\newblock {\em arXiv preprint arXiv:1601.06179}, 2016.

\bibitem{Ormerod2016a}
C.~M. Ormerod and E.~M. Rains.
\newblock A symmetric difference-differential lax pair for
  painlev$\backslash$'e vi.
\newblock {\em arXiv preprint arXiv:1603.04393}, 2016.

\bibitem{Gramani:Isomonodromic}
V.~G. Papageorgiou, F.~W. Nijhoff, B.~Grammaticos, and A.~Ramani.
\newblock Isomonodromic deformation problems for discrete analogues of
  {P}ainlev\'e equations.
\newblock {\em Phys. Lett. A}, 164(1):57--64, 1992.

\bibitem{Praagman:Solutions}
C.~Praagman.
\newblock Fundamental solutions for meromorphic linear difference equations in
  the complex plane, and related problems.
\newblock {\em J. Reine Angew. Math.}, 369:101--109, 1986.

\bibitem{rains:noncomgeom}
E.~M. Rains.
\newblock The noncommutative geometry of elliptic difference equations.
\newblock {\em to appear}.

\bibitem{rains:isomonodromy}
E.~M. Rains.
\newblock An isomonodromy interpretation of the hypergeometric solution of the
  elliptic {P}ainlev\'e equation (and generalizations).
\newblock {\em SIGMA Symmetry Integrability Geom. Methods Appl.}, 7:Paper 088,
  24, 2011.

\bibitem{Rains2013}
E.~M. Rains.
\newblock Generalized {H}itchin systems on rational surfaces.
\newblock {\em arXiv preprint arXiv:1307.4033}, 2013.

\bibitem{Rains2009}
E.~M. Rains and V.~P. Spiridonov.
\newblock Determinants of elliptic hypergeometric integrals.
\newblock {\em Functional Analysis and Its Applications}, 43(4):297--311, 2009.

\bibitem{Sakai:Rational}
H.~Sakai.
\newblock Rational surfaces associated with affine root systems and geometry of
  the {P}ainlev\'e equations.
\newblock {\em Comm. Math. Phys.}, 220(1):165--229, 2001.

\bibitem{Sakai:Garnier}
H.~Sakai.
\newblock A {$q$}-analog of the {G}arnier system.
\newblock {\em Funkcial. Ekvac.}, 48(2):273--297, 2005.

\bibitem{Sauloy}
J.~Sauloy.
\newblock {Galois} theory of {F}uchsian {$q$}-difference equations.
\newblock {\em Ann. Sci. \'Ecole Norm. Sup. (4)}, 36(6):925--968 (2004), 2003.

\bibitem{vanderPutSinger}
M.~van~der Put and M.~F. Singer.
\newblock {\em {Galois} theory of difference equations}, volume 1666 of {\em
  Lecture Notes in Mathematics}.
\newblock Springer-Verlag, Berlin, 1997.

\bibitem{VanderPut2003}
M.~van~der Put and M.~F Singer.
\newblock {\em Galois theory of linear differential equations}, volume 328.
\newblock Springer Science \& Business Media, 2003.

\bibitem{Yamada2009}
Y.~Yamada et~al.
\newblock A lax formalism for the elliptic difference painlev{\'e} equation.
\newblock {\em SIGMA}, 5(042):15, 2009.

\end{thebibliography}

\end{document}